\documentclass{sig-alternate}
\usepackage{graphicx}
\usepackage{amsmath,amssymb}
\usepackage{epsfig}
\usepackage{algorithm}
\usepackage{algorithmic}
\usepackage{url}
\usepackage{hyperref}

\usepackage{balance}

\newcommand{\spara}[1]{\smallskip\noindent{\bf #1}}

\newtheorem{definition}{Definition}

\DeclareMathOperator*{\argmax}{arg\,max}

\newcommand{\SBCN}{\textsf{SBCN}}

\newcommand{\squishlist}{
 \begin{list}{$\bullet$}
  {  \setlength{\itemsep}{0pt}
     \setlength{\parsep}{3pt}
     \setlength{\topsep}{3pt}
     \setlength{\partopsep}{0pt}
     \setlength{\leftmargin}{2em}
     \setlength{\labelwidth}{1.5em}
     \setlength{\labelsep}{0.5em}
} }
\newcommand{\squishlisttight}{
 \begin{list}{$\bullet$}
  { \setlength{\itemsep}{0pt}
    \setlength{\parsep}{0pt}
    \setlength{\topsep}{0pt}
    \setlength{\partopsep}{0pt}
    \setlength{\leftmargin}{2em}
    \setlength{\labelwidth}{1.5em}
    \setlength{\labelsep}{0.5em}
} }

\newcommand{\squishdesc}{
 \begin{list}{}
  {  \setlength{\itemsep}{0pt}
     \setlength{\parsep}{3pt}
     \setlength{\topsep}{3pt}
     \setlength{\partopsep}{0pt}
     \setlength{\leftmargin}{1em}
     \setlength{\labelwidth}{1.5em}
     \setlength{\labelsep}{0.5em}
} }

\newcommand{\squishend}{
  \end{list}
}

\graphicspath{{./fig/}}

\newcommand{\DAG}{\textsc{dag}}

\newcommand{\Probab}[1]{\mathcal{P}({#1})}
\newcommand{\Pcond}[2]{\Probab{{#1}\mid{#2}}}
\newcommand{\Pconj}[2]{\Probab{{#1} \wedge {#2}}}

\begin{document}
\title{Exposing the Probabilistic Causal Structure of Discrimination}

\numberofauthors{1}
\author{
Francesco Bonchi$^{1,2}$ \hspace{3mm}  Sara Hajian$^{2}$ \hspace{3mm}  Bud Mishra$^{3}$ \hspace{3mm}  Daniele Ramazzotti$^{4}$\\ \\
\begin{tabular}{cccc}
\affaddr{$^1$Algorithmic Data Analytics Lab}  & \affaddr{$^2$EURECAT} & \affaddr{$^3$New York University} &   \affaddr{$^4$Milano-Bicocca University}\\
\affaddr{ISI Foundation, Turin, Italy} & \affaddr{Barcelona, Spain} & \affaddr{New York, USA} & \affaddr{Milan, Italy}\\
 \sf{francesco.bonchi@isi.it} & \sf{sara.hajian@eurecat.org} &  \sf{mishra@nyu.edu} &  \sf{daniele.ramazzotti@disco.unimib.it}
\end{tabular}
}

\maketitle
\sloppy

\begin{abstract}

\emph{Discrimination discovery} from data is an important task aiming at identifying patterns of illegal and unethical discriminatory activities against protected-by-law groups, e.g.,~ethnic minorities. While any legally-valid proof of discrimination requires evidence of causality, the state-of-the-art methods are essentially correlation-based, albeit, as it is well known, correlation does not imply causation.

In this paper we take a principled causal approach to the data mining problem of discrimination detection in databases. Following Suppes' \emph{probabilistic causation theory}, we define a method to extract, from a dataset of historical decision records, the causal structures existing among the attributes in the data. The result is a type of constrained Bayesian network, which we dub \emph{Suppes-Bayes Causal Network} (\SBCN).
Next, we develop a toolkit of methods based on random walks on top of the \SBCN,  addressing different anti-discrimination legal concepts, such as direct and indirect discrimination, group and individual discrimination,  genuine requirement, and favoritism. Our experiments on real-world datasets confirm the inferential power of our approach in all these different tasks.
\end{abstract}

\section{Introduction} \label{sec:introduction}
%\enlargethispage*{2\baselineskip}

\spara{The importance of discrimination discovery.}
At the beginning of 2014, as an answer to the growing concerns about the role played by data mining algorithms in decision-making, USA President Obama called for a 90-day review of big data collecting and analysing practices. The resulting report\footnote{\url{http://www.whitehouse.gov/sites/default/files/docs/big_data_privacy_report_may_1_2014.pdf}} concluded that \emph{``big data technologies can cause societal harms beyond damages to privacy"}. In particular, it expressed concerns about the possibility that decisions informed by big data could have discriminatory effects, even in the absence of discriminatory intent, further imposing less favorable treatment to already disadvantaged groups.
%It further expressed alarm about the threats of an ``opaque decision-making environment'' guided by an ``impenetrable set of algorithms''.

Discrimination refers to an unjustified distinction of individuals based on their membership, or perceived membership, in a certain group or category. Human rights laws prohibit discrimination on several grounds, such as gender, age, marital status, sexual orientation, race, religion or belief, membership in a national minority, disability or illness. Anti-discrimination authorities (such as equality enforcement bodies, regulation boards, consumer advisory councils) monitor, provide advice, and report on discrimination compliances based on investigations and inquiries. A fundamental role in this context is played by \emph{discrimination discovery in databases}, i.e.,  the data mining problem of unveiling discriminatory practices by analyzing a dataset of historical decision records.

%Discrimination discovery is a fundamental task in understanding past and current trends of discrimination, in judicial dispute resolution in legal trials, in the validation of micro-data or of aggregated data before they are publicly released.

\spara{Discrimination is causal.}
According to current legislation, discrimination occurs when a group is treated ``less favorably''~\cite{legislationequal} than others, or when ``a higher proportion of people not in the group is able to comply'' with a qualifying criterion~\cite{Europeandisc}. Although these definitions do not directly imply causation, as stated in~\cite{foster2004causation} all discrimination claims require plaintiffs to demonstrate a causal connection between the challenged outcome and a protected status characteristic. In other words, in order to prove discrimination, authorities must answer the counterfactual question: what would have happened to a member of a specific group (e.g., nonwhite), if he or she had been part of another group (e.g., white)?

``The Sneetches'', the popular satiric tale\footnote{\url{http://en.wikipedia.org/wiki/The_Sneetches_and_Other_Stories}} against  discrimination published in 1961 by Dr. Seuss, describes a society of  yellow creatures divided in two races: the ones with a green star on their bellies, and the ones without. The  Star-Belly Sneetches have some privileges that are instead denied to Plain-Belly Sneetches. There are, however, Star-On and Star-Off machines that can make a Plain-Belly into a Star-Belly, and viceversa. Thanks to these machines, the causal relationship between race and privileges can be clearly measured, because stars can be placed on or removed from any belly, and multiple outcomes can be observed for an individual. Therefore, we could readily answer the counterfactual question, saying with certainty what would have happened to a Plain-Belly Sneetch had he or she been a Star-Belly Sneetch.

In the real world however, proving discrimination episodes is much harder, as we cannot manipulate race, gender, or sexual orientation of an individual. This highlights the need to assess discrimination as a causal inference problem~\cite{dabady2004measuring} from a database of past decisions,  where causality can be inferred probabilistically.
Unfortunately, \emph{the state of the art of data mining methods for discrimination discovery in databases
 does not properly address the causal question}, as it is mainly based on correlation-based methods (surveyed in Section \ref{sec:related_work}).

\spara{Correlation is not causation.} %\enlargethispage*{2\baselineskip}
It is well known that correlation between two variables does not necessarily imply that one causes the other.
Consider a unique cause $X$ of two effects, $Y$ and $Z$: if we do not take into account $X$, we might derive wrong conclusions because of the observable correlation  between $Y$ and $Z$. In this situation, $X$ is said to act as a \emph{confounding factor} for the relationship between $Y$ and $Z$.

Variants of the complex relation just discussed can arise even if, in the example, $X$ is not the actual cause of either $Y$ or $Z$, but it is only correlated to them, for instance, because of how the data were collected.
Consider for instance a credit dataset where there exists high correlation between a variable representing low income and another variable representing loan denial and let us assume that this is due to an actual legitimate causal relationship in the sense that, legitimately, a loan is denied if the applicant has low income. Let us now assume that we also observe high correlation between low income and being female, which, for instance, can be due to the fact that the women represented in the specific dataset in analysis, tend to be underpaid.
Given these settings, in the data we would also observe high correlation between the variable gender being female and the variable representing loan denial, due to the fact that we do not account for the presence of the variable low income. Following common terminologies, we will say that such situations are due to \emph{spurious correlations}.

However, the picture is even more complicated: it could be the case, in fact, that being female is the actual cause of the low income and, hence, be the \emph{indirect cause} of loan denial \emph{through} low income. This would represent a causal relationship between the gender and the loan denial, that we would like to detect as discrimination.

Disentangling these two different cases, i.e., female is only correlated to low income in a spurious way, or being female is the actual cause of low income, is at the same time important and challenging. This highlights the need for a principled causal approach to discrimination detection.

Another typical pitfall of correlation-based reasoning is expressed by what is known as Simpson's paradox\footnote{\url{http://en.wikipedia.org/wiki/Simpson's_paradox}} according to which, correlations observed in different groups might disappear when these heterogeneous groups are aggregated, leading to \emph{false positives} cases of discrimination discovery. One of the most famous false-positive examples due to Simpson's paradox occurred when in 1973 the University of California, Berkeley was sued for discrimination against women who had applied for admission to graduate schools. In fact, by looking at the admissions of 1973, it first appeared that men applying were significantly more likely to be admitted than women. But later, by examining the individual departments carefully, it was discovered that none of them was significantly discriminating against women. On the contrary, most departments had exercised a small bias in favor of women. The apparent discrimination was due to the fact that women tended to apply to departments with lower rates of admission, while men tended to apply to departments with higher rates~\cite{Bickel1975}. Later in Section \ref{subsec:exp_comp} we will use the dataset from this episode to highlight the differences between correlation-based and causation-based methods.
 %\enlargethispage*{2\baselineskip}

Spurious correlations can also lead to \emph{false negatives} (i.e., discrimination existing but not being detected) as is commonly seen in \emph{``reverse-discrimination''}. The typical case is when authorities take affirmative actions, e.g.,  with compensatory quota systems, in order to protect a minority group from a potential discrimination. Such actions, while
trying to erase the supposed discrimination (i.e., the spurious correlation), fail to address the real underlying causes for discrimination, potentially ending up denying individual members of a privileged group from access to their rightful shares of social goods. In the early 70's, a case involving the University of California at Davis Medical School highlighted one such incident as the school's admissions program reserved 16 of the 100 places in its entering class for ``disadvantaged" applicants, thus unintentionally reducing the chances of admission for a qualified applicant.\footnote{\url{http://en.wikipedia.org/wiki/Regents_of_the_University_of_California_v._Bakke}}

These are just few typical examples of the pitfalls of correlation-based reasoning in the discovery of discrimination. Later in Section \ref{subsec:exp_comp} we show concrete examples from real-world datasets where correlation-based methods to discrimination discovery are not satisfactory.

\spara{Our proposal and contributions.} %\enlargethispage*{2\baselineskip}
In this paper we take a principled causal approach to the data mining problem of discrimination detection in databases. Following Suppes' \emph{probabilistic causation theory}~\cite{probabilistic_causation,suppes_prima_facie} we define a method to extract, from a dataset of historical decision records, the causal structures existing among the attributes in the data.

In particular, we define the \emph{Suppes-Bayes Causal Network} (\SBCN), i.e., a directed acyclic graph (\DAG) where we have a node representing a Bernulli variable of the type $\langle attribute = value\rangle$ for each pair attribute-value present in the database. In this \DAG\ an arc $(A,B)$ represents the existence of a causal relation between $A$ and $B$ (i.e., $A$ causes $B$). Moreover, each arc is labeled with a score, representing the strength of the causal relation.

Our \SBCN\ is a constrained Bayesian network reconstructed by means of maximum likelihood estimation (MLE) from the given database, where we force the conditional probability distributions induced by the reconstructed graph to obey Suppes' constraints: i.e., \emph{temporal priority} and \emph{probability rising}.
Imposing Suppes' temporal priority and probability raising we obtain what we call the \emph{prima facie causes} graph~\cite{suppes_prima_facie},
which might still contain spurious causes (false positives). In order to remove these spurious case we add a bias term to the likelihood score, favoring sparser causal networks: in practice we sparsify the \emph{prima facie causes} graph by extracting a minimal set of edges which best explain the data. This regularization is done by means of the Bayesian Information Criterion (BIC)~\cite{bic_1978}.
%The biased estimator can be thought of as an empirical Bayes estimator.

The obtained \SBCN\ provides a clear summary, amenable to visualization, of the probabilistic causal structures found in the data. Such structures can be used to reason about different types of discrimination. In particular, we show how using several random-walk-based methods, where the next step in the walk is chosen proportionally to the edge weights, we can address different anti-discrimination legal concepts.
%This makes \SBCN\ a very general tool for discrimination detection.
Our experiments show that the measures of discrimination produced by our methods are very strong, almost binary, signals: our measures are very clearly separating the discrimination and the non-discrimination cases.

%\medskip

\emph{To the best of our knowledge this is the first proposal of discrimination detection in databases grounded in probabilistic causal theory.}

%\medskip

\spara{Roadmap.} The rest of the paper is organized as follows.
Section \ref{sec:related_work} discusses the state of the art in discrimination detection in databases.
%Section \ref{sec:discrimination_discovery} provide some preliminary concepts and needed background on discrimination detection and Suppes' probabilistic causation theory, which are at the basis of our method.
In Section \ref{sec:approach} we formally introduce the \SBCN\ and we present the method for extracting such causal network from the input dataset. Once extracted our  \SBCN, in Section \ref{sec:exploiting} we show how to exploit it for different concepts of discrimination detection, by means of random-walk methods.  Finally Section \ref{sec:experiments} presents our experimental assessment and comparison with correlation-based methods on two real-world datasets.

%\pagebreak

\section{Related work} \label{sec:related_work}
%\enlargethispage*{2\baselineskip}
Discrimination analysis is a multi-disciplinary problem, involving sociological causes, legal reasoning, economic models, statistical techniques~\cite{Custers2012,RR2013}. Some authors~\cite{HajianFerrer12,Kamiran2012} study how to prevent data mining from becoming itself a source of discrimination. In this paper instead we focus on the data mining problem of detecting discrimination in a dataset of historical decision records, and in the rest of this section we present the most related literature.

%More recently, the issue of anti-discrimination has been considered from a data mining perspective. Some proposals are oriented to using data mining to measure and discover discrimination~\cite{RPT2010}; while other proposals~\cite{HajianFerrer12,Kamiran2012} deal with preventing data mining from becoming itself a source of discrimination. Surveys of discrimination-aware data mining are in~\cite{Custers2012,RR2013}.

Pedreschi et al.~\cite{peder2008,pederruggi2009,RPT2010} propose a technique based on extracting classification rules (inductive part) and ranking the rules according to some legally grounded measures of discrimination (deductive part). The result is a (possibly large) set of classification rules, providing local and overlapping niches of possible discrimination. This model only deals with group discrimination.

 Luong et al.~\cite{Ruggieri2011} exploit the idea of \emph{situation-testing} \cite{rorive2009proving} to detect individual discrimination. For each member of the protected group with a negative decision outcome, testers with similar characteristics ($k$-nearest neighbors) are considered. If there are significantly different decision outcomes between the testers of the protected group and the testers of the unprotected group, the negative decision can be ascribed to discrimination.
%Similarity is modeled via a distance function. Testers are searched for among the $k$-nearest neighbors, and the difference is measured by some legally grounded measures of discrimination calculated over the two sets of testers. After this kind of labeling, a global description of those labeled as discriminated against can be extracted as a standard classification task. A real case study in the context of the evaluation of scientific projects for funding is presented by Romei et al.~\cite{Romei2012}.

Zliobaite et al.~\cite{Calders2011} focus on  the concept of \emph{genuine requirement} to detect that part of discrimination which may be explained by other, legally grounded, attributes. In~\cite{Dwork2012} Dwork et al. address the problem of fair classification that achieves both group fairness, i.e., the proportion of members in a protected group receiving positive classification is identical to the proportion in the population as a whole, and individual fairness, i.e., similar individuals should be treated similarly.

The above approaches assume that the dataset under analysis contains attributes that denote protected groups (i.e., direct discrimination). This may not be the case when such attributes are not available, or not even collectable at a micro-data level as in the case of the loan applicant's race. In these cases we talk about indirect discrimination discovery.
Ruggieri et al.~\cite{pedreschiICAIL,atttack2014} adopt a form of rule inference to cope with the indirect discovery of discrimination. The correlation information is called background knowledge, and is itself coded as an association rule.
%Moreover, in~\cite{Ruggieri2014} authors deploy privacy attack strategies as tools for discrimination discovery.
%The intuition comes from the intriguing parallel between the role of the anti-discrimination authority in the three scenarios above and the role of an attacker in private data publishing. Moreover, the problem of achieving simultaneous discrimination prevention and privacy protection in data publishing and mining was recently addressed in~\cite{Hajian2014}. They design strategies and algorithms inspired/based on Fr`echet bounds attacks, attribute inference attacks, and minimality attacks to the purpose of unveiling hidden discriminatory practices.

Mancuhan and Clifton~\cite{Mancuhan2014} propose Bayesian networks as a tool for discrimination discovery. Bayesian networks consider the dependence between all the attributes and use these dependencies in estimating the joint probability distribution without any strong assumption, since a Bayesian network graphically represents a factorization of the joint distribution in terms of conditional probabilities encoded in the edges.
Although Bayesian networks are often used to represent causal relationships, this needs not be the case, in fact a directed edge from two nodes of the network does not imply any causal relation between them. As an example, let us observe that the two graphs $A \rightarrow B \rightarrow C$ and $C \rightarrow B \rightarrow A$ impose exactly the same conditional independence requirements and, hence, any Bayesian network would not be able to disentangle the direction of any causal relationship among these events.

%\pagebreak

Our work departs from this literature as:
\begin{enumerate}
  \item it is grounded in probabilistic causal theory instead of being based on correlation;
  \item  it proposes a holistic approach able to deal with different types of discrimination in a single unifying framework, while the methods in the state of the art usually deal with one and only one specific type of discrimination;
  \item it is the first work to adopt graph theory and social network analysis concepts, such as random-walk-based centrality measures and community detection, for discrimination detection;
\end{enumerate}
 Our proposal has also lower computational cost than methods such as~\cite{peder2008,pederruggi2009,RPT2010} which require to compute a potentially exponential number of association/classification rules.

\section{Suppes-Bayes Causal Network} \label{sec:approach}
%\enlargethispage*{\baselineskip}
%In~\cite{Mancuhan2014}, discrimination is discovered in terms of a significant difference in the Bayesian elift between the original Bayesian network and the relative one. In particular, this comparison aims at disentangling situations of genuine discriminations from spurious correlations. In this work, we aim at improving this approach by explicitly dealing with spurious correlations with a method grounded in the theory of probabilistic causation. To do so, we will analyze the discriminatory phenomena by means of causal Bayesian networks in place of the general Bayesian networks.

 In order to study discrimination as a causal inference problem, we exploit the criteria defined in the theories of \emph{probabilistic causation}~\cite{probabilistic_causation}. In particular, we follow~\cite{suppes_prima_facie}, where Suppes proposed the notion of \emph{prima facie causation} that is at the core of probabilistic causation.
Suppes' definition is based on two pillars: $(i)$ any cause must happen before its effect (\emph{temporal priority}) and $(ii)$ it must raise the probability of observing the effect (\emph{probability raising}).

%\medskip
%\newdef{definition}{Definition}
\begin{definition}[Probabilistic causation~\cite{suppes_prima_facie}] \label{def:praising}
\emph{For any two events $h$ and $e$, occurring respectively at times $t_h$ and $t_e$, under the  mild assumptions that $0 < \Probab{h}, \Probab{e} < 1$, the event $h$ is called a prima facie cause of the event $e$ if it occurs before the effect and the cause raises the probability of the effect, i.e.} $t_h < t_e \quad \text{and} \quad \Pcond{e}{h} > \Pcond{e}{\neg h} \,.
$
\end{definition}

In the rest of this section we introduce our method to construct, from a given relational table $\mathcal{D}$,  a type of causal Bayesian network constrained to satisfy the conditions dictated by Suppes' theory, which we dub \emph{Suppes-Bayes Causal Network} (\SBCN).

In the literature many algorithms exist to carry out structural learning of general Bayesian networks and they usually fall into two families~\cite{koller2009probabilistic}. The first family, {\em constraint based learning}, explicitly tests for pairwise independence of variables conditioned on the power set of the rest of the variables in the network. These algorithms exploit structural conditions defined in various approaches to causality~\cite{probabilistic_causation,counterfactual_causation,manipulability_causation}. The second family, {\em score based learning}, constructs a network which maximizes the likelihood of the observed data with some regularization constraints to avoid overfitting.
Several hybrid approaches have also been recently proposed~\cite{brenner2013sparsityboost,caprese_causation,capri_causation}.

Our framework can be considered a hybrid approach exploiting \emph{constrained maximum likelihood estimation} (MLE) as follows: $(i)$ we first define all the possible causal relationship among the variables in $\mathcal{D}$ by considering only the oriented edges between events that are consistent with Suppes' notion of probabilistic causation and, subsequently, $(ii)$ we perform the reconstruction of the \SBCN\ by a score-based approach (using BIC), which considers only the valid edges.

%The idea of adopting Suppes' theory to reconstruct the causal structure subsumed by a progression model was introduced for the first time in~\cite{caprese_causation,capri_causation}, albeit in a completely different context, i.e., modelling somatic evolution in cancer.

%\medskip
%
%In particular, following~\cite{capri_causation}, in the first phase of our approach, we exploit the two ingredients of Suppes' definition, i.e., the cause must happen before its effect (\emph{temporal priority}) and it must raise the probability of observing the effect (\emph{probability raising}) as follow.
We next present in details the whole learning process.

\subsection{Suppes' constraints}
\label{sec:sc}
We start with an input relational table $\mathcal{D}$ defined over a set $A$ of $h$ categorical attributes and $s$ samples. In case continuous numerical attributes exists in $\mathcal{D}$, we assume they have been discretized to become categorical. From $\mathcal{D}$, we derive $\mathcal{D^{\prime}}$, an $m \times s$ binary matrix representing
$m$ Bernoulli variables of the type $\langle attribute = value\rangle$, where an entry is $1$ if we have an observation for the specific variable and $0$ otherwise.

\spara{Temporal priority.}
The first constraint, temporal priority, cannot be simply checked in the data as we have no timing information for the events. In particular, in our context the events for which we want to reason about temporal priority are the Bernoulli variables $\langle attribute = value\rangle$.

The idea here is that, e.g.,  $income=low$ cannot be a cause of $gender=female$, because the time when the gender of an individual is determined is antecedent to that of when the income is determined. This intuition is implemented by simply letting the data analyst provide as input to our framework a partial temporal order $r: A \rightarrow \mathbb{N}$ for the $h$ attributes, which is then inherited from the $m$ Bernoulli variables \footnote{Note that our learning technique requires the input order $r$ to be correct and complete in order to guarantee its convergence. Nevertheless, if this is not the case, it is still capable of providing valuable insights about the underlying causal model, although with the possibility of false positive or false negative causal claims.}.

Based on the input dataset $\mathcal{D}$ and the partial order $r$ we produce the first graph $G=(V,E)$ where we have a node for each of the Bernoulli variables, so $|V| = m$, and we have an arc $(u,v) \in E$ whenever $r(u) \leq r(v)$. This way we will immediately rule out causal relations that do not satisfy the temporal priority constraint.

\spara{Probability raising.} Given the graph $G=(V,E)$ built as described above the next step requires to prune the arcs which do not satisfy the second constraint, probability raising, thus building $G'=(V,E')$, where $E' \subseteq E$. In particular we remove from $E$ each arc $(u,v)$ such that $\Pcond{v}{u} \leq \Pcond{v}{\neg u}$. The graph $G'$ so obtained is called  \emph{prima facie} graph.

We recall that the probability raising condition is equivalent to constraining for positive statistical dependence~\cite{caprese_causation}: in the prima facie graph we model \emph{all and only} the positive correlated relations among the nodes already partially ordered by temporal priority, consistently with Suppes' characterization of causality in terms of relevance.

%Given a node $u \in V$ let $Pa(u) = \{v \in V | (v,u) \in E\}$ denote the set of parents of $u$. Suppes' probability raising constraint requires the conditional probability for each pair $u, Pa(u)$ in the network to be:
%\[
%\begin{cases}
%\Pcond{u}{Pa(u)} = \theta \\
%\Pcond{u}{\~{Pa(u)}} \le \epsilon
%\end{cases}
%\]
%where $\theta, \epsilon \in [0,1]$ and $\theta \gg \epsilon$. This makes us further constraints on the valid edges which leads us to consider $G^{\prime} \subset G$ by removing from $G$ all the edges where the probability raising is not verified.

\begin{figure*}[t!]
\centering
\vspace{-4mm}
\begin{tabular}{c}
\hspace{-10mm}\includegraphics[width=1.1\linewidth]{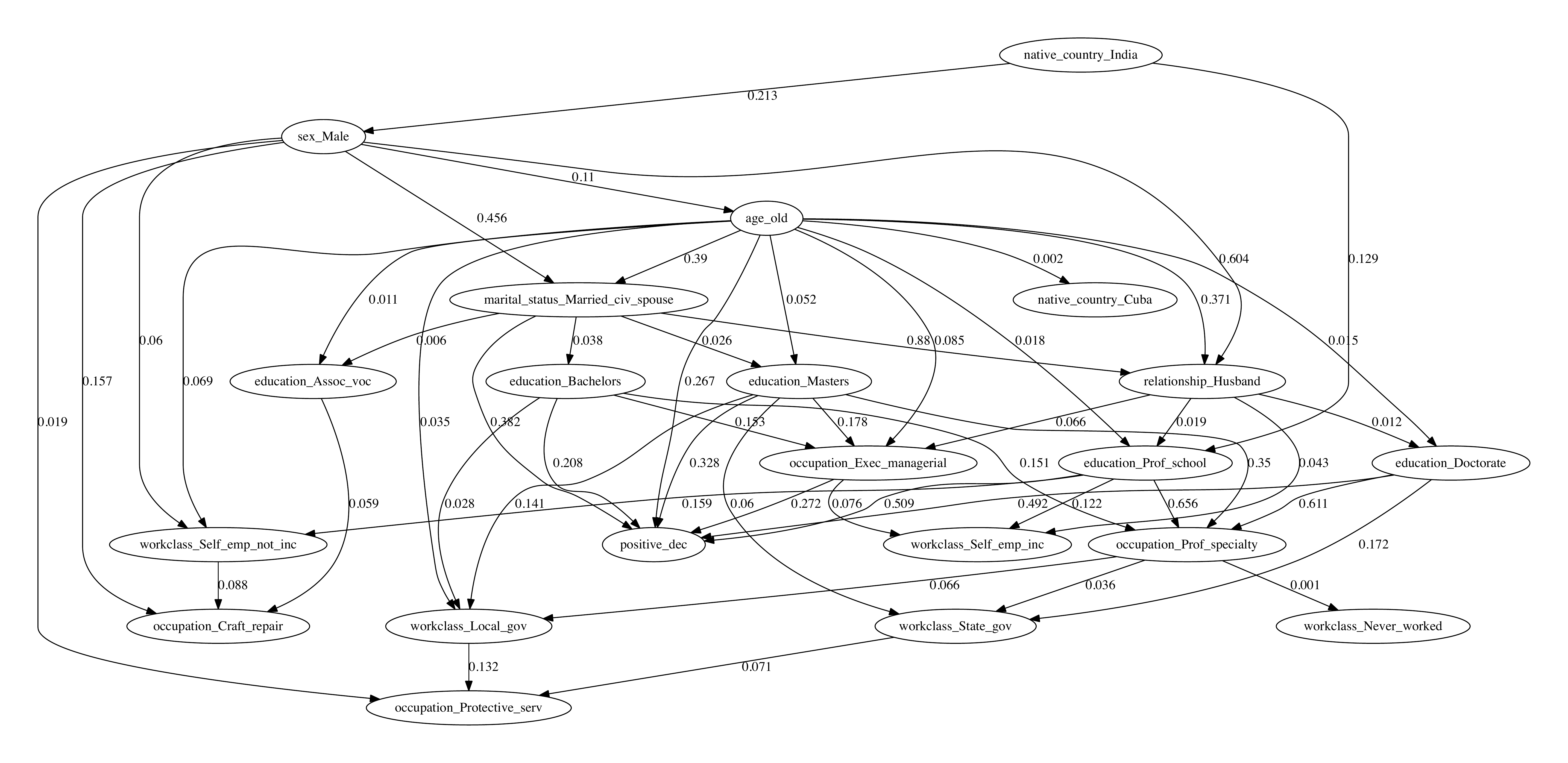}
\end{tabular}
\vspace{-6mm}
\caption{One portion of the \SBCN\ extracted from the \textsf{Adult} dataset. This subgraph corresponds to the $C_2$ community reported later in Table~\ref{tab:communities} (Section \ref{sec:experiments}) extracted by a community detection algorithm. \label{fig:example}}
\vspace{-2mm}
\end{figure*}

\pagebreak
\subsection{Network simplification}
Suppes' conditions are necessary but not sufficient to evaluate causation~\cite{capri_causation}: especially when the sample size is small, the model may have false positives (spurious causes), even after constraining for Suppes' temporal priority and probability raising criteria (which aim at removing false negatives). Consequently, although we expect all the statistically relevant causal relations to be modelled in $G'$, we also expect some spurious ones in it.

In our proposal, in place of other structural conditions used in various approaches to causality, (see e.g.,~\cite{probabilistic_causation,counterfactual_causation,manipulability_causation}), we perform a network simplification (i.e., we sparsify the network by removing arcs) with a score based approach, specifically by relying on the Bayesian Information Criterion (BIC) as the regularized likelihood score~\cite{bic_1978}.

We consider as inputs for this score the graph $G^{\prime}$ and the dataset $\mathcal{D^{\prime}}$. Given these, we select the set of arcs $E^* \subseteq E'$ that maximizes the score:
\begin{equation*}\label{eq:bic}
score_{\text{BIC}}(\mathcal{D^{\prime}},G^{\prime}) = LL(\mathcal{D^{\prime}} | G^{\prime}) - \dfrac{\log s}{2} \text{dim}(G^{\prime}).
\end{equation*}
In the equation, $G^{\prime}$ denotes the graph, $\mathcal{D^{\prime}}$ denotes the data, $s$ denotes the number of samples, and $\text{dim}(G^{\prime})$ denotes the number of parameters in $G^{\prime}$. Thus, the regularization term $-\text{dim}(G^{\prime})$ favors graphs with fewer arcs. The coefficient ${\log s}/{2}$ weighs the regularization term, such that the higher the weight, the more sparsity will be favored over ``explaining" the data through maximum likelihood. Note that the likelihood is implicitly weighted by the number of data points, since each point contributes to the score.

Assume that there is one \emph{true} (but unknown) probability distribution that generates the observed data, which is, eventually, uniformly randomly corrupted by false positives and negatives rates (in $[0,1)$). Let us call \emph{correct model}, the statistical model which best approximate this distribution. The use of BIC on $G'$ results in removing the false positives and, asymptotically (as the sample size increases), converges to the correct model. In particular, BIC is attempting to select the candidate model corresponding to the highest Bayesian Posterior probability, which can be proved to be equivalent to the presented score and its $log(s)$ penalization factor.

We denote with $G^* = (V,E^*)$ the graph that we obtain after this step. We note that, as for general Bayesian network, $G^*$ is a \DAG\ by construction.

\subsection{Confidence score}
Using the reconstructed \SBCN, we can represent the probabilistic relationships between any set of events (nodes). As an example, suppose to consider the nodes representing respectively $income=low$ and $gender=female$ being the only two direct causes (i.e., with arcs toward) of  $loan=denial$. Given \SBCN, we can estimate the conditional probabilities for each node in the graph, i.e., probability of $loan=denial$ given $income=low$ $AND$ $gender=female$ in the example, by computing the conditional probability of only the pair of nodes directly connected by an arc. For an overview of state-of-the-art methods for doing this, see~\cite{koller2009probabilistic}. However, we expect to be mostly dealing with full data, i.e., for every directly connected node in the \SBCN, we expect to have several observations of any possible combination $attribute=value$. For this reason, we can simply estimate the node probabilities by counting the observations in the data. Moreover, we will exploit such conditional probabilities to define the confidence score of each arc in terms of their causal relationship.

In particular, for each arc $(v,u) \in E^*$ involving the causal relationship between two nodes $u,v \in V$, we define a confidence score $W(v,u) = \Pcond{u}{v} - \Pcond{u}{\neg v}$, which, intuitively, aims at estimating the observations where the cause $v$ is followed by its effect $u$, that is $\Pcond{u}{v}$, and the ones where this is not observed, i.e., $\Pcond{u}{\neg v}$, because of imperfect causal regularities. We also note that, by the constraints discussed above, we require $\Pcond{u}{v} \gg \Pcond{u}{\neg v}$ and, for this reason, each weight is positive and no larger than 1, i.e., $W: E^* \rightarrow (0,1]$.

Combining all of the concepts discussed above, we conclude with the following definition.

%\medskip

\begin{definition}[Suppes-Bayes Causal Network] \label{def:scn}\enlargethispage*{\baselineskip}
\emph{Given an input dataset $\mathcal{D^{\prime}}$ of $m$ Bernoulli variables and $s$ samples, and given a partial order $r$ of the variables, the Suppes-Bayes Causal Network $SBCN = (V,E^*,W)$ subsumed by $\mathcal{D^{\prime}}$ is a weighted \DAG\ such that the following requirements hold:}
\squishlist
\item \emph{\textbf{[Suppes' constraints]} for each arc $(v,u) \in E^*$ involving the causal relationship between nodes $u,v \in V$, under the  mild assumptions that $0 < \Probab{u}, \Probab{v} < 1$}:
$$
r(v) \leq r(u) \quad \text{and} \quad \Pcond{u}{v} > \Pcond{u}{\neg v} \,.
$$
\item \emph{\textbf{[Simplification]} let $E'$ be the set of  arcs satisfying the Suppes' constraints as before; among all the subsets of $E'$, the set of arcs $E^*$ is the one whose corresponding graph maximizes BIC:}
$$
E^* = \argmax_{E \subseteq E', G =(V,E)}  (LL(\mathcal{D^{\prime}} | G) - \dfrac{\log s}{2} \text{dim}(G)) \,.
$$

\item \emph{\textbf{[Score]}} $W(v,u) = \Pcond{u}{v} - \Pcond{u}{\neg v},\; \forall (v,u) \in E^*$
\squishend
\end{definition}

An example of a portion of a \SBCN\ extracted from a real-world dataset is reported in Figure \ref{fig:example}.

Algorithm \ref{algo:sbcn} summarizes the learning approach adopted for the inference of the \SBCN\ .
Given $\mathcal{D^{\prime}}$ an input dataset over $m$ Bernoulli variables and $s$ samples, and $r$ a partial order of the variables, Suppes' constraints are verified (Lines $4$-$9$) to construct a \DAG\ as described in Section~\ref{sec:sc}.

The likelihood fit is performed by hill climbing (Lines $12$-$21$), an iterative optimization technique that starts with an arbitrary solution to a problem (in our case an empty graph) and then attempts to find a better solution by incrementally visiting the neighbourhood of the current one. If the new candidate solution is better than the previous one it is considered in place of it. The procedure is repeated until the stopping criterion is matched.

The $!StoppingCriterion$ occurs (Line $12$) in two situations: $(i)$ the procedure stops when we have performed a large enough number of iterations or, $(ii)$ it stops when none of the solutions in $G_{neighbors}$ is better than the current $G_{fit}$.
Note that $G_{neighbors}$ denotes all the solutions that are derivable from $G_{fit}$ by  removing or adding at most one edge.

\begin{algorithm}
\caption{{Learning the \SBCN\ }}
\label{algo:sbcn}
\begin{algorithmic}[1]
\STATE{Inputs: $\mathcal{D^{\prime}}$ an input dataset of $m$ Bernoulli variables and $s$ samples, and $r$ a partial order of the variables}
%\STATE{Being $r$ a partial order of the variables.}
\STATE{Output: $SBCN(V,E^*,W)$ as in Definition 2}
\STATE{\textbf{[Suppes' constraints]}}
\FORALL{pairs $(v,u)$ among the $m$ Bernoulli variables}
\IF{$r(v) \leq r(u)$ \AND $\Pcond{u}{v} > \Pcond{u}{\neg v}$}
\STATE{add the arc $(v,u)$ to $SBCN$.}
\ENDIF
\ENDFOR
\STATE{\textbf{[Likelihood fit by hill climbing]}}
\STATE{Consider $G(V,E^*,W)_{fit} = \emptyset$.}
\WHILE{$!StoppingCriterion()$}
\STATE{Let $G(V,E^*,W)_{neighbors}$ be the neighbor solutions of $G(V,E^*,W)_{fit}$.}
\STATE{Remove from $G(V,E^*,W)_{neighbors}$ any solution whose arcs are not included in $SBCN$.}
\STATE{Consider a random solution $G_{current}$ in $G(V,E^*,W)_{neighbors}$.}
\IF{$score_{BIC}(\mathcal{D^{\prime}},G_{current})>score_{BIC}(\mathcal{D^{\prime}},G_{fit})$}
\STATE{$G_{fit} = G_{current}$.}
\STATE{$\forall$ arc $(v,u)$ of $G_{fit}$,  $W(v,u) = \Pcond{u}{v} - \Pcond{u}{\neg v}$.}
\ENDIF
\ENDWHILE
\STATE{$SBCN = G_{fit}$.}
\RETURN $SBCN$.
\end{algorithmic}
\end{algorithm}

%$!StoppingCriterion()$ indicates that the procedure goes on untill either we have performed a big enough number of iterations or untill none of the solutions in $G_{neighbors}$ are better than the current $G_{fit}$.
%Moreover, by $G_{neighbors}$ we mean all the solutions that are derivable from $G_{fit}$ by 1 change at most (i.e., by removing or adding exactly one edge).

\spara{Time and space complexity.}
The computation of the valid \DAG\ according to Suppes' constraints (Lines $4$-$10$) requires a pairwise calculation among the $m$ Bernoulli variables. After that, the likelihood fit by hill climbing (Lines $11$-$21$) is performed. Being an heuristic, the computational cost of hill climbing depends on the sopping criterion. However, constraining by Suppes' criteria tends to regularize the problem leading on average to a quick convergence to a good solution. The time complexity of Algorithm~\ref{algo:sbcn} is $\mathcal{O}(s m)$ and the space required is $\mathcal{O}(m^2)$, where $m$ however is usually not too large, being the number of attribute-value pairs, and not the number of examples.

\pagebreak
\subsection{Expressivity of a SBCN} %\enlargethispage*{\baselineskip}
We conclude this Section with a discussion on the causal relations that we model by a $SBCN$.

Let us assume that there is one true (but unknown) probability distribution that generates the observed data whose structure can be modelled by a \DAG . Furthermore, let us consider the causal structure of such a \DAG\ and let us also assume each node with more then one cause to have conjunctive parents: any observation of the child node is preceded by the occurrence of all its parents. As before we call correct model, the statistical model which best approximate the distribution. On these settings, we can prove the following theorem.

%\smallskip

\newtheorem{theorem}{Theorem}
\begin{theorem}
\emph{Let the sample size $s \to \infty$, the provided partial temporal order $r$ be correct and complete and the data be uniformly randomly corrupted by false positives and negatives rates (in $[0,1)$), then the $SBCN$ inferred from the data is the correct model.}
\end{theorem}

\begin{proof} \emph{[Sketch]}
Let us first consider the case where the observed data have no noise. On such an input, we observe that the prima facie graph has no false negatives: in fact $\forall [c \to e]$ modelling a genuine causal relation, $\Pconj{e}{c} = \Probab{e}$, thus the probability raising constraint is satisfied, so it is the temporal priority given that we assumed  $r$ to be correct and complete.

Furthermore, it is know that the likelihood fit performed by $BIC$ converges to a class of structures equivalent in terms of likelihood among which there is the correct model: all these topologies are the same unless the directionality of some edges. But, since we started with the prima facie graph which is already ordered by temporal priority, we can conclude that in this case the $SBCN$ coincides with the correct model.

To extend the proof to the case of data uniformly randomly corrupted by false positives and negatives rates (in $[0,1)$), we note that the marginal and joint probabilities change monotonically as a consequence of the assumption that the noise is uniform. Thus, all inequalities used in the preceding proof still hold, which concludes the proof.
\end{proof}

%\smallskip

 %\enlargethispage*{\baselineskip}

In the more general case of causal topologies where any cause of a common effect is independent from any other cause (i.e., we relax the assumption of conjunctive parents), the $SBCN$ is not guaranteed to converge to the correct model but it coincides with a subset of it modeling all the edges representing statistically relevant causal relations (i.e., where the probability raising condition is verified).

\section{Discrimination discovery by random walks}\label{sec:exploiting}
%\enlargethispage*{2\baselineskip}
In this section we propose several random-walk-based methods over the reconstructed \SBCN, to deal with different discrimination-detection tasks.

%We show how anti-discrimination authorities can exploit the \SBCN\ structure to reason about the legal concepts of
%direct discrimination, indirect discrimination, favoritism, genuine
%requirement, and situation-testing (individual discrimination).

\subsection{Group discrimination and favoritism}\label{subsec:groupdiscr}
The basic problem in the analysis of direct discrimination is precisely to quantify the degree
of discrimination suffered by a given protected group (e.g., an ethnic group) with respect to a decision (e.g., loan denial). In contrast to discrimination, favoritism refers to the case of an individual treated better than others for reasons
not related to individual merit or business necessity: for instance,
favoritism in the workplace might result in a person being promoted faster than others
unfairly. In the following we denote favoritism as positive discrimination in contrast with negative discrimination.

Given an \SBCN\ we define a measure of  group discrimination (either negative or positive) for each node $v \in V$.
Recall that each node represents a pair $\langle attribute = value\rangle$, so it is essentially what we refer to as a group, e.g., $\langle gender = female\rangle$.
Our task is to assign a score of discrimination $ds^{-}: V \rightarrow [0, 1]$ to each node, so that the closer $ds^-(v)$ is to 1 the more discriminated is the group represented by $v$.

We compute this score by means of a number $n$ of random walks that start from $v$ and reaches either the node representing the positive decision or the one representing the negative decision.
In these random walks the next step is chosen proportionally to the weights of the out-going arcs.
Suppose a random walk has reached a node $u$, and let $deg_{out}(u)$ denote the set of outgoing arcs from $u$.
Then the arc $(u,z)$ is chosen with probability $$p(u,z) = \frac{W(u,z)}{\sum_{e \in deg_{out}(u)} W(e)}.$$
When a random walk ends in a node with no outgoing arc before reaching either the negative or the positive decision, it is restarted from the source node $v$.

%\medskip

\begin{definition}[Group discrimination score]\label{def:gds}
\emph{Given an $SBCN=(V,E^*,W)$, let $\delta^- \in V$ and $\delta^+ \in V$ denote the nodes indicating the negative and positive decision, respectively.
Given a node $v \in V$, and a number $n \in \mathbb{N}$ of random walks to be performed, we denote as $rw_{v \rightarrow \delta^-}$ the number of random walks started at node $v$ that reach  $\delta^-$ earlier than $\delta^+$. The discrimination score for the group corresponding to node $v$ is then defined as
$$
ds^-(v)= \frac{rw_{v \rightarrow \delta^-}}{n}.
$$
This implicitly also defines a score of positive discrimination (or favoritism): $ds^+(v) = 1 - ds^-(v)$.}
\end{definition}

%\medskip

Taking advantage of  the \SBCN\, we also propose two additional measures capturing how far a node representing a group is from the positive and negative decision respectively.
This is done by computing the average number of steps that the random walks take to reach the two decisions: we denote these scores  as $as^{-}(v)$ and $as^{+}(v)$.

\subsection{Indirect discrimination}
The European Union Legislation \cite{Europeandisc} provides a broad definition of indirect discrimination as occurring ``where an apparently neutral provision, criterion or practice would put persons of a racial or ethnic origin at a particular disadvantage compared with other persons". In other words, the actual result of the apparently neutral provision is the same as an explicitly discriminatory one. A typical legal case study of indirect discrimination is concerned with \emph{redlining}: e.g., denying a loan because of ZIP code, which in some areas is an attribute highly correlated to race. Therefore, even if the attribute race cannot be required at loan-application time (thus would not be present in the data), still race discrimination is perpetrated.
Indirect discrimination discovery refers to the data mining task of discovering the attributes values that can act as a proxy to the protected groups and lead to discriminatory decisions indirectly \cite{peder2008,pederruggi2009,HajianFerrer12}.

In our setting, indirect discrimination can be detected by applying the same method described in Section~\ref{subsec:groupdiscr}.

\subsection{Genuine requirement}
 The legal concept of genuine requirement refers to detecting that part of the discrimination which may be explained by other, legally-grounded, attributes; e.g., denying credit to women may be explainable by the fact that most of them have low salary or delay in returning previous credits. A typical example in the literature is the one of the ``genuine occupational requirement", also called ``business necessity" in \cite{USdisc,ellis2012eu}. In the state of the art of data mining methods for discrimination discovery, it is also known as \emph{explainable discrimination} \cite{hajian2014discrimination} and \emph{conditional discrimination} \cite{Calders2011}.

The task here is to evaluate to which extent the discrimination apparent for a group is ``explainable'' on a legal ground.
Let $v \in V$ be the node representing the group which is suspected of being discriminated, and $u_l \in V$ be a node whose causal relation with a negative or positive decision is legally grounded. As before, $\delta^-$ and $\delta^+$ denote the negative and positive decision, respectively.
Following the same random-walk process described in Section~\ref{subsec:groupdiscr}, we define the \emph{fraction of explainable discrimination} for the group $v$:
$$
fed^{-}(v)= \frac{rw_{v \rightarrow u_l \rightarrow  \delta^-}}{rw_{v \rightarrow \delta^-}},
$$
i.e.,  the fraction of random walks passing trough  $u_l $ among the ones started in $v$ and reaching  $\delta^-$ earlier than $\delta^+$. Similarly we define $fed^{+}(v)$, i.e., the fraction of explainable positive discrimination.

\pagebreak
\subsection{Individual and subgroup discrimination} \label{sec:MID} %\enlargethispage*{\baselineskip}
Individual discrimination requires to measure the amount of discrimination for a specific individual, i.e., an entire record in the database. Similarly, subgroup discrimination refers to discrimination against a subgroup described by a combination of multiple protected and non-protected attributes: personal data, demographics, social, economic and cultural indicators, etc. For example, consider the case of gender discrimination in credit approval: although an analyst may observe that no discrimination occurs in general, it may turn out that older women obtain car loans only rarely.

Both problems can be handled by generalizing the technique introduced in Section~\ref{subsec:groupdiscr} to deal with a set of starting nodes, instead of only one.
Given an $SBCN=(V,E^*,W)$ let $v_1,\ldots,v_n$ be the nodes of interest. In order to define a discrimination score for $v_1,\ldots,v_n$, we perform a \emph{personalized PageRank}\cite{PPR1} computation with respect to $v_1,\ldots,v_n$.
In personalized PageRank, the probability of jumping to a node when abandoning the random walk is not uniform, but it is given by a vector of probabilities for each node.
In our case the vector will have the value $\frac{1}{n}$ for each of the nodes $v_1,...,v_n \in V$ and zero for all the others.
The output of personalized PageRank is a score $ppr(u |v_1,...,v_n)$ of proximity/relevance to  $\{v_1,...,v_n \}$ for each other node $u$ in the network. In particular, we are interested in the score of the nodes representing the negative and positive decision: i.e., $ppr(\delta^- |v_1,...,v_n)$ and $ppr(\delta^+ |v_1,...,v_n)$ respectively.

%As done in Definition \ref{def:gds} we normalize the two scores so to sum to 1.

%\medskip

\begin{definition}[Generalized discrimination score]\label{def:gends}
\emph{Given an $SBCN=(V,E^*,W)$, let $\delta^- \in V$ and $\delta^+ \in V$ denote the nodes indicating the negative and positive decision, respectively.
Given a set of nodes $v_1,...,v_n \in V$, we define the generalized (negative) discrimination score for the subgroup or the individual represented by $\{v_1,...,v_n \}$ as
$$
gds^-(v_1,...,v_n)= \frac{ppr(\delta^- |v_1,...,v_n)}{ppr(\delta^- |v_1,...,v_n) + ppr(\delta^+ |v_1,...,v_n)}.
$$
This implicitly also defines a generalized score of positive discrimination: $gds^+(v_1,...,v_n) = 1 - gds^-(v_1,...,v_n)$.}
\end{definition}

\section{Experimental Evaluation} \label{sec:experiments}
%\enlargethispage*{2\baselineskip}

This section reports the experimental evaluation of our approach on four datasets, \emph{Adult}, \emph{German credit} and \emph{census-income} from the UCI Repository of Machine Learning Databases\footnote{\url{http://archive.ics.uci.edu/ml}}, and \emph{Berkeley Admissions Data} from \cite{Feedman1998}. These are well-known real-life datasets typically used in discrimination-detection literature.

\medskip

 \textsf{Adult:}  consists of 48,842 tuples and 10 attributes, where each tuple correspond to an individual and it is described by personal attributes such as age, race, sex, relationship, education, employment, etc. Following the literature, in order to define the decision attribute we use the income levels, $\leq$50K (negative decision) or $>$50K (positive decision). We use four levels in the partial order for temporal priority: age, race, sex, and native country are defined in the first level; education, marital status, and relationship are defined in the second level; occupation and work class are defined in the third class, and the decision attribute (derived from income) is  the last level.

 \medskip

 \textsf{German credit:} consists of 1000 tuples with 21 attributes on bank account holders applying for credit. The decision attribute is based on repayment history, i.e., whether the customer is labeled with good or bad credit risk. Also for this dataset the partial order for temporal priority has four orders. Personal attributes such as gender, age, foreign worker are defined in the first level. Personal attributes such as employment status and job status are defined in the second level.  Personal properties such as savings status and credit history are defined in the third level, and finally the decision attribute is the last level.

\medskip

 \textsf{Census-income:} consists of 299,285 tuples and 40 attributes, where each tuple correspond to an individual and it is described by demographic and employment attributes such as age, sex, relationship, education, employment, ext. Similar to  \textsf{Adult} dataset, the decision attribute is the income levels and we define four levels in the partial order for temporal priority.

% This data set contains weighted census data extracted from the 1994 and 1995 Current Population Surveys conducted by the U.S. Census Bureau. The data contains 41 demographic and employment related variables.

\medskip

%\smallskip
Building the \SBCN\ just take a handful of seconds in \textsf{German credit} and \textsf{Adult}, and few minutes in \textsf{Census-income} on a commodity laptop.
The main characteristics of the extracted  \SBCN\ are reported in Table~\ref{tab:characteristics}, while the distribution of the edges scores $W(e)$ is plotted in Figure \ref{fig:edgeprobsNew}.
\begin{table}[t!]
\small
\caption{\SBCN\ main characteristics.\label{tab:characteristics}}
\centering \vspace{1mm}
\begin{tabular}{c c c c c c}
{\em Dataset} & $|V|$  & $|A|$ &   {\em avgDeg} & {\em maxInDeg} & {\em maxOutDeg}
\\
\hline
\textsf{Adult} & 92  & 230  &  2.5 &   7 &  19  \\
\textsf{German credit} & 73  &  102  & 1.39  &  3  & 7  \\
\textsf{Census-income} &  386  & 1426 & 3.69  &  8  & 54  \\
%\aline
%\hline
%\textsf{Census-income} & 386 & 1426 & 3.69 &  8 &  54 \\
\hline
\end{tabular}
\end{table}
\begin{figure}[t!]
%\vspace{-2mm}
\begin{tabular}{c}
\hspace{-12mm}\includegraphics[width=1.2\linewidth]{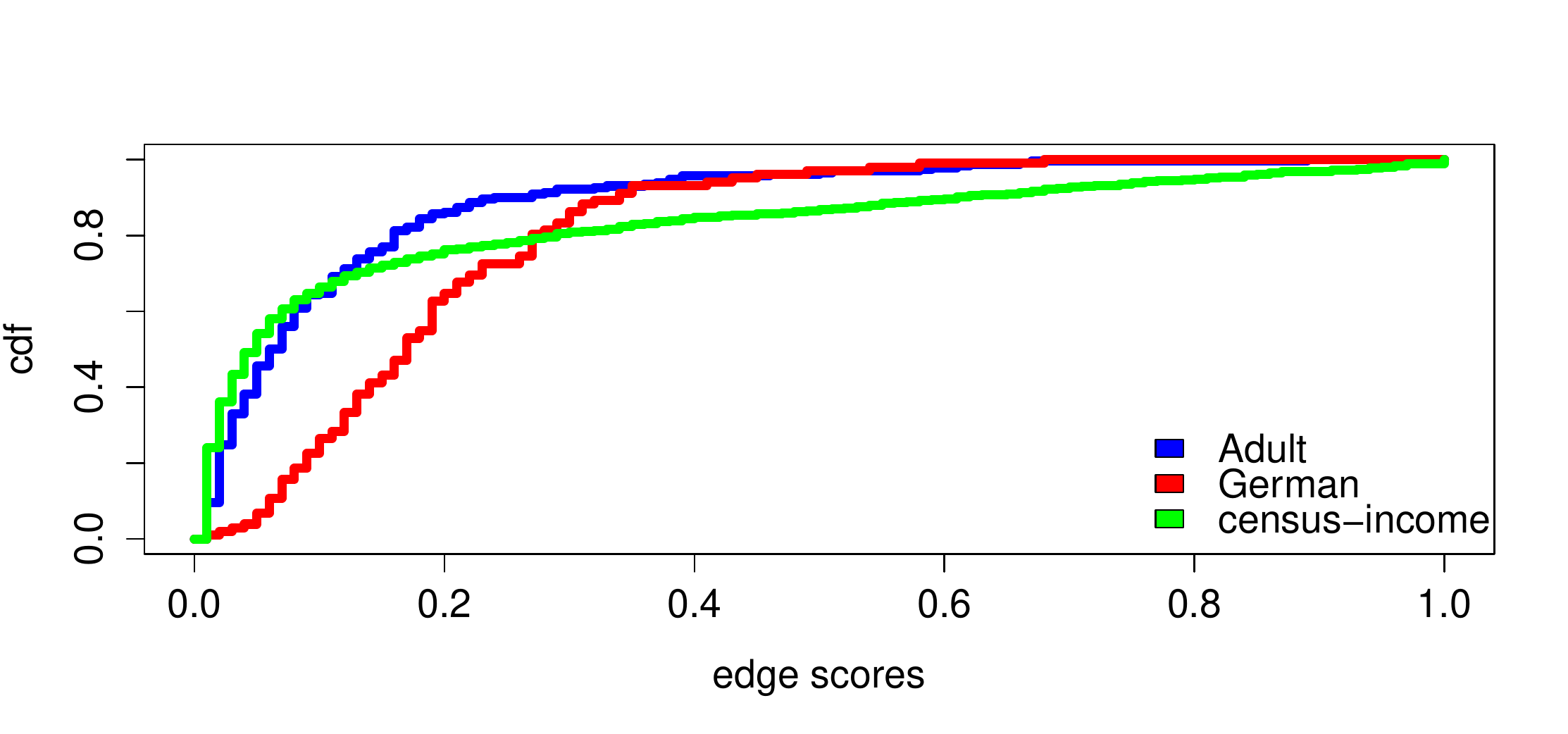}
\end{tabular}
%\vspace{-6mm}
\caption{Distribution of edge scores. \label{fig:edgeprobsNew}}
%\vspace{3mm}
\end{figure}

\smallskip

As discussed in the Introduction we also use the dataset from the famous 1973 episode at University of California at Berkeley, in order to highlight the differences between correlation-based and causation-based methods.

\medskip

\textsf{Berkeley Admissions Data:} consists of 4,486 tuples and three attributes, where each tuple correspond to an individual and it is described by the gender of applicants and the department that they apply for it. For this dataset the partial order for temporal priority has three orders. Gender is defined in the first level, department in the second level, and finally the decision attribute is the last level.
Table~\ref{tab:BAD} is a three-way table that presents admissions data at the University of California, Berkeley in 1973 according to the variables department (A, B, C, D, E), gender (male, female), and outcome (admitted, denied). The table is adapted from data in the text by Freedman, et al. \cite{Feedman1998}.

\begin{table}[t!]
\centering
\begin{tabular}{llrllr}
\hline
\multicolumn{2}{c}{Male} & \multicolumn{2}{c}{Female}  \\
\cline{1-4}
Admitted   & Denied	 & Admitted & Denied & Department \\
\hline
512      & 313   &  89 &19  & A  \\
313      & 207   & 17   &    8  & B   \\
120     & 205    & 202  & 391 & C    \\
138     & 279    & 131 & 244  &  D   \\
53 &  138    & 94 &  299 &  E   \\
22 & 351 & 24 & 317 & F  \\
\hline
\end{tabular}
\caption{Berkeley Admission Data}
\label{tab:BAD}
%\vspace{3mm}
\end{table}

\subsection{Community detection on the \SBCN}
Given that our \SBCN\ is a directed graph with edge weight, as a first characterization we try to partition it using a random-walks-based community detection algorithm, called \emph{Walktrap} and proposed in \cite{pons2005computing}, whose unique parameter is the maximum number of steps in a random walk (we set it to 8), and which automatically identifies the right number of communities.  The idea is that short random walks tend to stay in the same community (densely connected area of the graph). Using this algorithm over the reconstructed \SBCN\ from \textsf{Adult} dataset, we obtain $5$ communities: two larger ones and three smaller ones (reported in Table~\ref{tab:communities}).
Interestingly, the two larger communities seem built around the negative ($C_1$)  and the positive ($C_2$) decisions.

Figure~\ref{fig:example} in Section \ref{sec:approach} shows the subgraph of the \SBCN\ corresponding to $C_2$ (that we can call, the favoritism cluster): we note that such cluster also contains nodes such as \textsf{sex\_Male}, \textsf{age\_old}, \textsf{relationship\_Husband}. The other large community $C_1$, can be considered the discrimination cluster: beside the negative decision it contains other nodes representing disadvantaged groups such as \textsf{sex\_Female}, \textsf{age\_young}, \textsf{race\_Black}, \textsf{marital\_status\_Never\_married}.
This good separability of the \SBCN\ in the two main clusters of discrimination and favoritism, highlights the goodness of the causal structure captured by the \SBCN.
\begin{table}[t!]
\caption{Communities found in the \SBCN\ extracted from the \textsf{Adult} dataset by  \emph{Walktrap}\cite{pons2005computing}. In the table the attributes are shortened as in parenthesis: age (ag), education (ed), marital\_status (ms), native\_country (nc), occupation (oc), race(ra), relationship (re), sex (sx), workclass (wc).}
\vspace{2mm}
\label{tab:communities}

\centering
\begin{scriptsize}
\begin{tabular}{c}
\hline
$C_1$ \\
\hline
\textbf{negative\_dec},  wc:Private, ed:Some\_college, ed:Assoc\_acdm,\\
 \textbf{ms:Never\_married}, ms:Divorced, \textbf{ms:Widowed},\\
  ms:Married\_AF\_spouse, oc:Sales, oc:Other\_service,\\
 oc:Priv\_house\_serv,  re:Own\_child, re:Not\_in\_family, re:Wife,\\
  \textbf{re:Unmarried}, re:Other\_relative, \textbf{ra:Black}, oc:Armed\_Forces,\\
   oc:Handlers\_cleaners,  oc:Tech\_support, oc:Transport\_moving,\\
ed:7th\_8th, ed:10th, ed:12th, ms:Separated,\\
 ed:HS\_grad,ed:11th, nc:Outlying\_US\_Guam\_USVI\_etc,\\
   nc:Haiti, \textbf{ag:young},\textbf{ sx:Female},  ra:Amer\_Indian\_Eskimo,\\
    nc:Trinadad\_Tobago, nc:Jamaica, oc:Machine\_op\_inspct,\\
   ms:Married\_spouse\_absent, oc:Adm\_clerical, \\
\hline
$C_2$\\
\hline
\textbf{positive\_dec}, oc:Prof\_specialty, wc:Self\_emp\_not\_inc,\\
 ms:Married\_civ\_spouse, oc:Craft\_repair,oc:Protective\_serv,\\
 \textbf{re:Husband}, ed:Prof\_school, wc:Self\_emp\_inc,\\
 \textbf{ag:old} , wc:Local\_gov, \textbf{oc:Exec\_managerial},\\
  ed:Bachelors, ed:Assoc\_voc, ed:Masters, wc:Never\_worked,\\
 wc:State\_gov, ed:Doctorate, \textbf{sx:Male}, nc:India, nc:Cuba\\
\hline
$C_3$\\
 \hline
 oc:Farming\_fishing, wc:Without\_pay, nc:Mexico, nc:Canada,\\
nc:Italy, nc:Guatemala, nc:El\_Salvador, ra:White, \\
 nc:Poland, ed:1st\_4th, ed:9th,ed:Preschool, ed:5th\_6th\\
   \hline
$C_4$\\
   \hline
 nc:Iran, nc:Puerto\_Rico, nc:Dominican\_Republic,\\
nc:Columbia, nc:Peru, nc:Nicaragua,  ra:Other \\
\hline
$C_5$ \\
   \hline
  nc:Philippines, nc:Cambodia, nc:China, nc:South, \\
   nc:Japan, nc:Taiwan, nc:Hong, nc:Laos, nc:Thailand,\\
    nc:Vietnam, ra:Asian\_Pac\_Islander\\
\hline
\end{tabular}
\end{scriptsize}
\vspace{4mm}
\end{table}

\subsection{Group discrimination and favoritism}
We next focus on assessing the discrimination score $ds^{-}$  we defined in Section \ref{subsec:groupdiscr}, as well as the average number of steps that the random walks take to reach the negative and positive decisions, denote $as^{-}(v)$  and $as^{+}(v)$ respectively. %Recall that the favoritism (or positive discrimination) score is implicitly defined as $ds^+(v) = 1 - ds^-(v)$.

Tables~\ref{tab:adultdscores}, \ref{tab:germanscores} and \ref{tab:incomescores} report the top-5 and bottom-5 nodes w.r.t.  the discrimination score $ds^{-}$, for datasets \textsf{Adult}, \textsf{German} and \textsf{Census-income}, respectively. The first and most important observation is that our discrimination score provides a very clear signal, with some  disadvantaged groups having very high discrimination score (equal to 1 or very close), and similarly clear signals of favoritism, with groups having $ds^-(v) = 0$, or equivalently $ds^+(v) = 1$. This is more clear in the \textsf{Adult} dataset, where the positive and negative decisions are artificially derived from the income attribute. In the \textsf{German credit} dataset, which is more realistic as the decision attribute is truly about credit, both discrimination and favoritism are less palpable.
This is also due to the fact that \textsf{German credit} contains less proper causal relations, as reflected in the higher sparsity of the \SBCN. A consequence of this sparsity is also that the random walks generally need more steps to reach one of the two decisions.
In \textsf{Census-income} dataset, we observe favoritism with respect to married and asian\_pacific individuals.

\pagebreak
\subsection{Genuine requirement}
We next focus on genuine requirement (or explainable discrimination).
Table~\ref{tab:edsAdult}  reports some examples of fraction of explainable discrimination (both positive and negative) on the \textsf{Adult} dataset. We can see how some fractions of discrimination against protected groups, can be ``explained'' by intermediate nodes such as having a low education profile, or a simple job. In the case these intermediate nodes are considered legally grounded, then one cannot easily file a discrimination claim.%\enlargethispage*{4\baselineskip}

Similarly, we can observe that the favoritism towards groups such as married men, is  explainable, to a large extent, by higher education and good working position, such as managerial or executive roles.

\begin{table}[t!]
\caption{Top-5 and bottom-5 groups by discrimination score $ds^-(v)$ in \textsf{Adult} dataset. \label{tab:adultdscores}}
%\vspace{-2mm}
\centering
\begin{small}
\begin{tabular}{r|c c c|}
\multicolumn{1}{c}{} & \multicolumn{1}{c}{$ds^-(v)$}   &  $as^{-}(v)$ & \multicolumn{1}{c}{$as^{+}(v)$}\\
\cline{2-4}
\textsf{\scriptsize relationship\_Unmarried} &  1     & 1.164   &   -   \\
\textsf{\scriptsize marital\_status\_Never\_married} &  0.996    & 1.21  &  2.14 \\
\textsf{\scriptsize age\_Young} &  0.995   & 2.407   &  3.857  \\
\textsf{\scriptsize race\_Black} &  0.994   &  2.46 &   4.4   \\
\textsf{\scriptsize sex\_Female} &  0.98   &  2.60 &   3.76   \\
\cline{2-4}
\end{tabular}

\smallskip

\begin{tabular}{r|c c c|}
\multicolumn{1}{c}{\hspace{-8mm} }& \multicolumn{1}{c}{$ds^-(v)$}   &  $as^{-}(v)$ & \multicolumn{1}{c}{$as^{+}(v)$}\\
\cline{2-4}
\textsf{\scriptsize relationship\_Husband} &   0   & -   &   2   \\
\hspace{-8mm} \textsf{\scriptsize marital\_status\_Married\_civ\_spouse} & 0  & -  &  2.06 \\
\textsf{\scriptsize sex\_Male} &    0  &  -   &  3.002  \\
\textsf{\scriptsize native\_country\_India}   & 0.002   &  4.0 &   3.25   \\
\textsf{\scriptsize age\_Old} &   0.018    &  2.062 &   2.14   \\
\cline{2-4}
\end{tabular}
\end{small}

\bigskip

\caption{Top-5 and bottom-4 groups by discrimination score $ds^-(v)$ in \textsf{German credit}. We report only  the bottom-4, because there are only 4 nodes in which $ds^+(v) > ds^-(v)$. \label{tab:germanscores}}
%\vspace{-2mm}
\centering
\begin{small}
\begin{tabular}{r|c c c|}
\multicolumn{1}{c}{} & \multicolumn{1}{c}{$ds^-(v)$}   &  $as^{-}(v)$ & \multicolumn{1}{c}{$as^{+}(v)$}\\
\cline{2-4}
\textsf{\scriptsize residence\_since\_le\_1d6} &  1   & 6.0  &  -    \\
\textsf{\scriptsize residence\_since\_gt\_2d8} &  1  & 2.23  & -  \\
\hspace{-6mm}\textsf{\scriptsize residence\_since\_from\_1d6\_le\_2d2} &  1   &  6.0  &    -    \\
\textsf{\scriptsize age\_gt\_52d6} & 0.86   &  3.68 &  4.0    \\
\textsf{\scriptsize personal\_status\_male\_single} & 0.791   &  5.15 &  5.0    \\
\cline{2-4}
\end{tabular}

\smallskip

\begin{tabular}{r|c c c|}
\multicolumn{1}{c}{\hspace{-8mm} }& \multicolumn{1}{c}{$ds^-(v)$}   &  $as^{-}(v)$ & \multicolumn{1}{c}{$as^{+}(v)$}\\
\cline{2-4}
\textsf{\scriptsize job\_unskilled\_resident} &   0  & -  &  2.39    \\
\hspace{-6mm}\textsf{\scriptsize personal\_status\_male\_mar\_or\_wid} & 0.12  & 8.0  & 4.4 \\
\textsf{\scriptsize age\_le\_30d2} &  0.186  &  7.0  &    3.34     \\
\textsf{\scriptsize personal\_status\_female\_} &   0.294  &  6.48 &  4.4    \\
\cline{2-4}
\multicolumn{1}{c}{\textsf{\scriptsize div\_or\_sep\_or\_mar}} & \multicolumn{3}{c}{} \\
\end{tabular}
\end{small}

\bigskip

\caption{Top-5 and bottom-5 groups by discrimination score $ds^-(v)$ in \textsf{Census-income} dataset. \label{tab:incomescores}}
%\vspace{-2mm}
\centering
\begin{small}
\begin{tabular}{r|c c c|}
\multicolumn{1}{c}{} & \multicolumn{1}{c}{$ds^-(v)$}   &  $as^{-}(v)$ & \multicolumn{1}{c}{$as^{+}(v)$}\\
\cline{2-4}
\textsf{\scriptsize MIGSAME\_Not\_in\_universe\_under\_1\_year\_old} &  0.71  & 4.09   &  8.82  \\
\textsf{\scriptsize WKSWORK\_94\_5\_inf} &  0.625 & 3.0  &  6.76 \\
\textsf{\scriptsize AWKSTAT\_Not\_in\_labor\_force} &  0.59  &  2.0 &  6.16  \\
\textsf{\scriptsize VETYN\_0\_5\_20\_5} & 0.58 & 1.01   &  5.17 \\
\textsf{\scriptsize MARSUPWT\_3188\_455\_4277\_98} &  0.55 & 5.0 &  9.25   \\
\cline{2-4}
\end{tabular}

\smallskip

\begin{tabular}{r|c c c|}
\multicolumn{1}{c}{\hspace{-8mm} }& \multicolumn{1}{c}{$ds^-(v)$}   &  $as^{-}(v)$ & \multicolumn{1}{c}{$as^{+}(v)$}\\
\cline{2-4}
\textsf{\scriptsize AHGA\_Doctorate\_degreePhD\_EdD} &   0   & -   &  3.07   \\
\hspace{-8mm} \textsf{\scriptsize AMARITL\_Married\_A\_F\_spouse\_present} & 0  & -  &  4.49 \\
\textsf{\scriptsize AMJOCC\_Sales} &   0  &  -   &  2.0  \\
\textsf{\scriptsize ARACE\_Asian\_or\_Pacific\_Islander}   &  0 &  -  &  6.47 \\
\textsf{\scriptsize VETYN\_20\_5\_32\_5} &  0  &  - &   5.89   \\
\cline{2-4}
\end{tabular}
\end{small}

\end{table}

\begin{table}[t!]
\caption{Fraction of explainable discrimination for some exemplar pair of nodes in the \textsf{Adult} dataset. \label{tab:edsAdult}}

%\vspace{-1mm}

\centering
\begin{scriptsize}
\begin{tabular}{ccc}
\hspace{-2mm}Source node & Intermediate  & $fed^{-}(v)$ \\
\hline
\hspace{-2mm}\textsf{race\_Amer\_Indian\_Eskimo} &  \textsf{education\_HS\_grad}   &  $0.481$ \\
\hspace{-2mm}\textsf{sex\_Female} &  \textsf{occupation\_Other\_service}  &  $0.310$ \\
\hspace{-2mm}\textsf{age\_Young} &  \textsf{occupation\_Other\_service}  & $0.193$\\
\hspace{-2mm}\textsf{relationship\_Unmarried} & \textsf{education\_HS\_grad}  & $0.107$  \\
\hspace{-2mm}\textsf{race\_Black} &  \textsf{education\_11th}   &   $0.083$ \\
\hline
\end{tabular}

\bigskip

\begin{tabular}{ccc}
Source node & Intermediate  & $fed^{+}(v)$ \\
\hline
\textsf{relationship\_Husband} &  \textsf{occupation\_Exec\_managerial}  &   $0.806$  \\
\textsf{sex\_Male} &  \textsf{occupation\_Exec\_managerial} &  $0.587$ \\
\textsf{native\_country\_Iran} &   \textsf{education\_Bachelors}   &  $0.480$  \\
\textsf{native\_country\_India} &  \textsf{education\_Prof\_school}  &   $0.415$   \\
\textsf{age\_Old} &  \textsf{occupation\_Exec\_managerial}  &   $0.39$   \\
\hline
\end{tabular}
\end{scriptsize}
\vspace{-4mm}
\end{table}

%\begin{table*}[t!]
%\caption{Fraction of explainable discrimination for some exemplar pair of nodes in the \textsf{German credit} dataset. \label{tab:edsGerman}}
%\vspace{-2mm}
%\centering
%
%\begin{scriptsize}
%\begin{tabular}{cc}
%\begin{tabular}{ccc}
%Source node & Intermediate  & $fed^{-}(v)$ \\
%\hline
%\textsf{residence\_since\_le\_1d6} &  \textsf{credit\_amount\_from\_11154d4\_le\_14789d2}   &  $1$ \\
%\textsf{residence\_since\_gt\_2d8} &  \textsf{checking\_status\_lt\_0} &  $0.938$  \\
%\textsf{residence\_since\_from\_1d6\_le\_2d2} &   \textsf{credit\_amount\_from\_11154d4\_le\_14789d2}   &  $1$  \\
%\textsf{age\_gt\_52d6} & \textsf{checking\_status\_lt\_0}   &  $0.813$  \\
%\textsf{personal\_status\_male\_single} & \textsf{credit\_amount\_from\_11154d4\_le\_14789d2}  &  $0.290$ \\
%\hline
%\end{tabular}
%
%&
%
%\begin{tabular}{ccc}
%Source node & Intermediate  & $fed^{+}(v)$ \\
%\hline
%\textsf{job\_unskilled\_resident} &  \textsf{duration\_le\_17d6}  & $1$   \\
%\textsf{personal\_status\_male\_mar\_or\_wid} &   \textsf{job\_skilled} &  $0.328$  \\
%\textsf{age\_le\_30d2} &   \textsf{duration\_le\_17d6}  &  $1$     \\
%\textsf{personal\_status\_female\_div\_or\_dep\_or\_mar} & \textsf{job\_skilled}   &  $0.218$    \\
%\hline
%\end{tabular}
%\end{tabular}
%\end{scriptsize}
%\end{table*}

%\pagebreak

\subsection{Subgroup and Individual Discrimination}%\enlargethispage*{\baselineskip}
We next turn our attention to subgroup and individual discrimination discovery. Here the problem is to assign a score of discrimination not to a single node (a group), but to multiple nodes (representing the attributes of an individual or a subgroup of citizens). In Section \ref{sec:MID} we have introduced based on the PageRank of the positive and negative decision, $ppr(\delta^+)$ and $ppr(\delta^-)$ respectively, personalized on the nodes of interest.
Figure \ref{fig:scatt1} presents a scatter plot of $ppr(\delta^+)$ versus $ppr(\delta^-)$ for each individual in the \textsf{German credit} dataset. We can observe the perfect separation between individuals corresponding to a high personalized PageRank with respect to the positive decision, and those associated with a high personalized PageRank relative to the negative decision.

\begin{figure}[t!]
%\vspace{-3mm}
\centering
\includegraphics[width=.9\linewidth]{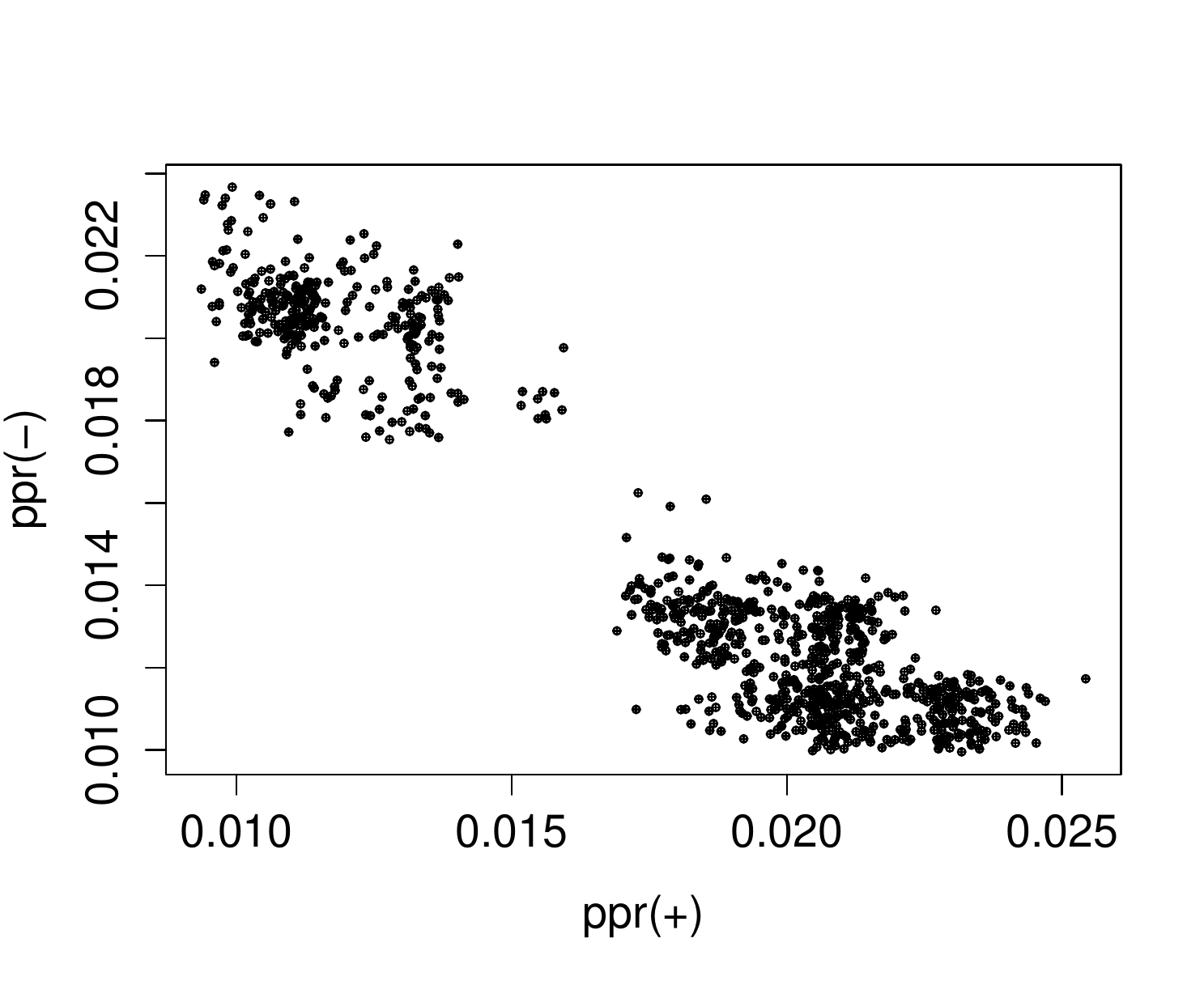}
\vspace{-2mm}
\caption{Scatter plot of $ppr(\delta^+)$ versus $ppr(\delta^-)$ for each individual in the \textsf{German credit} dataset.\label{fig:scatt1}}
\vspace{-1mm}
\end{figure}

Such good separation is also reflected in the \emph{generalized discrimination score} (Definition 4) that we obtain by combining $ppr(\delta^+)$ versus $ppr(\delta^-)$.

In Figure \ref{germanhisto} we report the distribution of the generalized discrimination score $gds^-$ for the population of the \textsf{German credit} dataset: we can make a note of the clear separation between the two subgroups of the population.

\begin{figure}[t!]
%\vspace{2mm}
\centering
\includegraphics[width=.9 \linewidth]{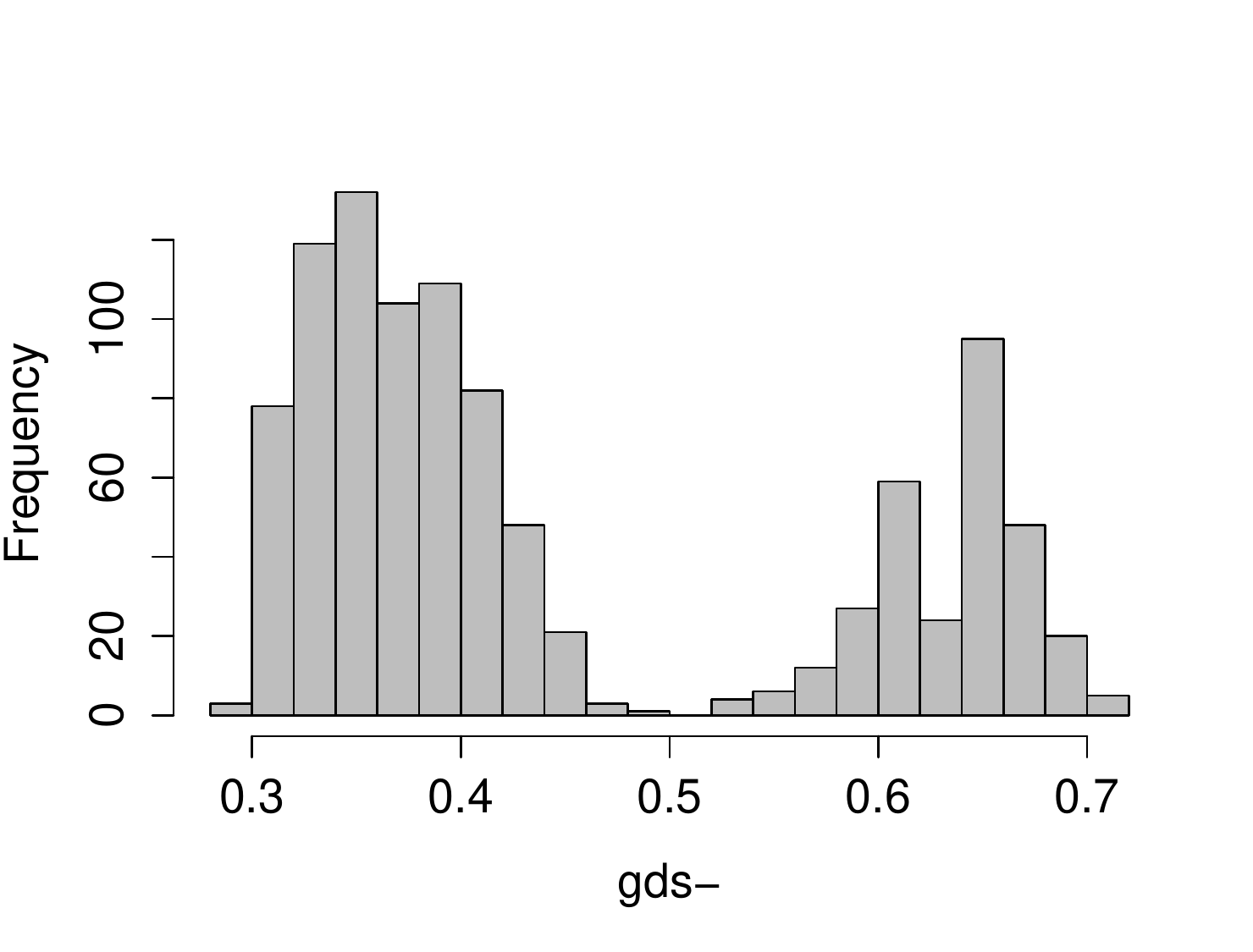}
\vspace{-4mm}
\caption{Individual discrimination: histogram representing the distribution of the values of the generalized discrimination score $gds^-$ for the population of the \textsf{German credit} dataset.\label{germanhisto}}
\vspace{2mm}
\end{figure}

In the \textsf{Adult} dataset (Figure \ref{fig:scatt2}) we do not observe the same neat separation in two subgroups as in the \textsf{German credit} dataset, also due to the much larger number of points. Nevertheless, as expected, $ppr(\delta^+)$ and $ppr(\delta^-)$ still exhibit anticorrelation. In Figure \ref{fig:scatt2} we also use colors to show two different groups:  red dots are for \textsf{age\_Young} and blue dots are for \textsf{age\_Old} individuals. As expected we can see that the red dots are distributed more in the area of higher $ppr(\delta^-)$.

\begin{figure}[h!]
\centering
\includegraphics[width=.9\linewidth]{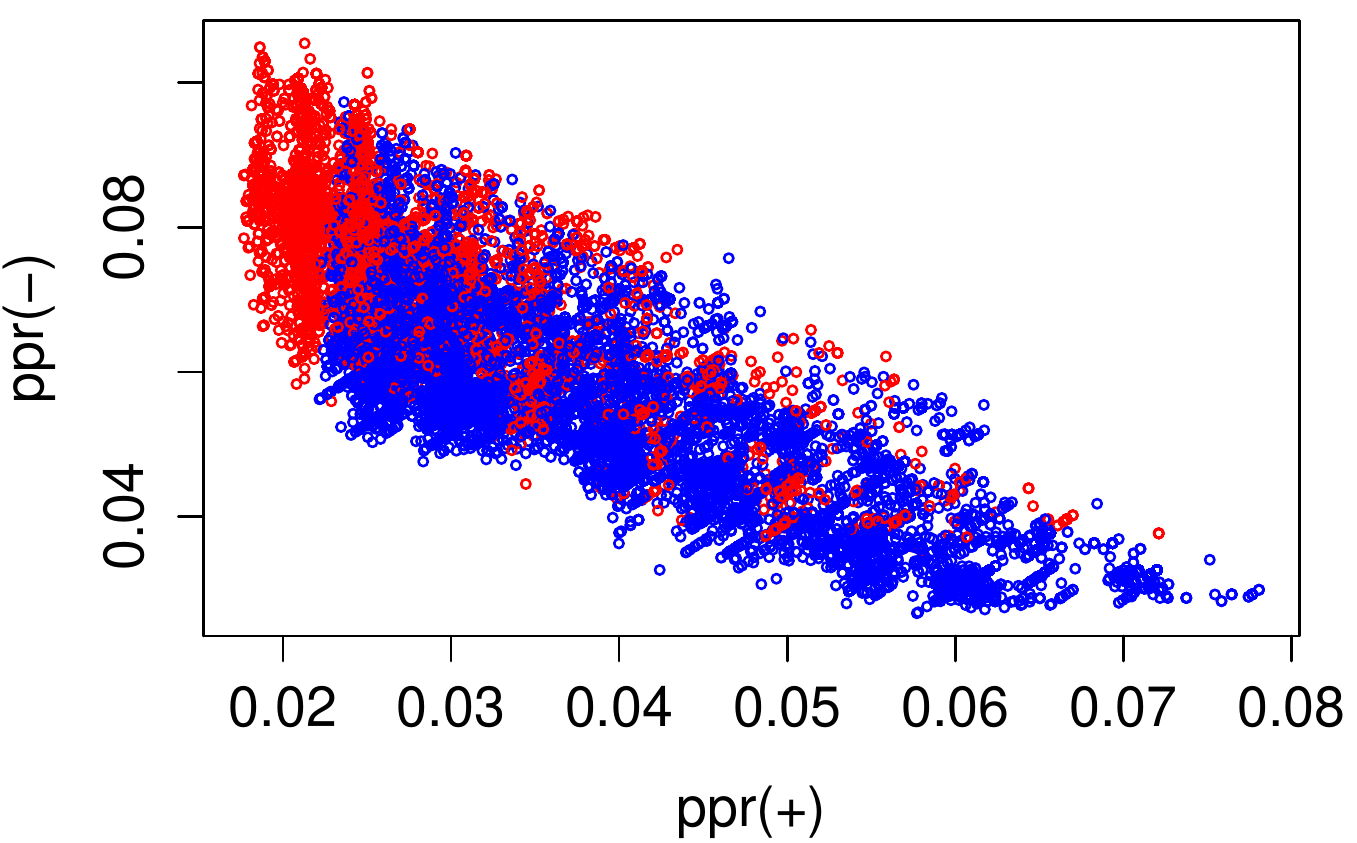}
\caption{Individual discrimination: scatter plot of $ppr(\delta^+)$ versus $ppr(\delta^-)$ for each individual in the \textsf{Adult} dataset. Red dots are for \textsf{age\_Young} and blue dots are for \textsf{age\_Old}. \label{fig:scatt2}}
\vspace{2mm}
\end{figure}

%\begin{figure}[htp]
%\centering
%\includegraphics[width=\linewidth]{fig/adultgdshistgray.pdf}
%\caption{Distribution of edge scores \label{fig:histoadult}}
%\end{figure}

The plots in Figure \ref{fig:multipleccdf}  have a threshold $t \in [0,1]$ on the X-axis, and the fraction of tuples having $gds^-() \geq t$ on the Y-axis, and they show this for different subgroups. The first plot, from the \textsf{Adult} dataset, shows the group female, young, and young female. As we can see the individuals that are both young and female have a higher  generalized discrimination score.
Similarly, the second plot shows the groups old, single male, and old single male from the \textsf{German credit} dataset. Here we can observe much lower rates of discrimination with only $1/5$ of the corresponding populations having $gds^-()  \geq 0.5$, while in the previous plot it was more than 85\%.

\begin{figure}[t!]
%\vspace{1mm}
\centering
\includegraphics[width=\linewidth]{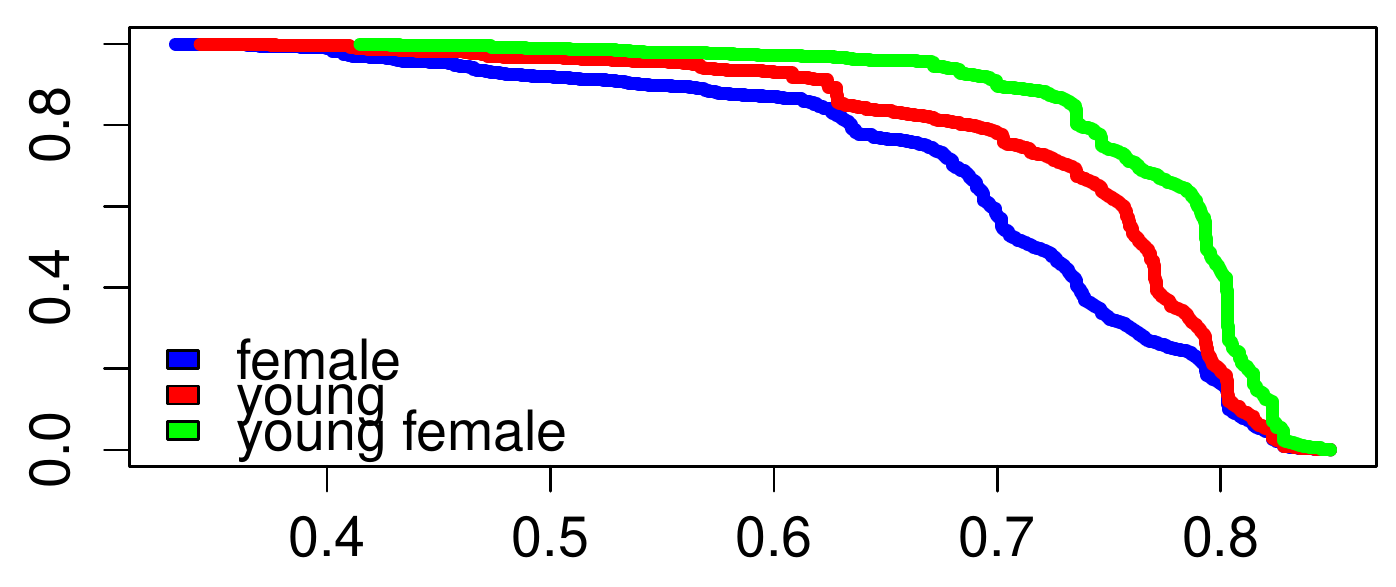}
\includegraphics[width=\linewidth]{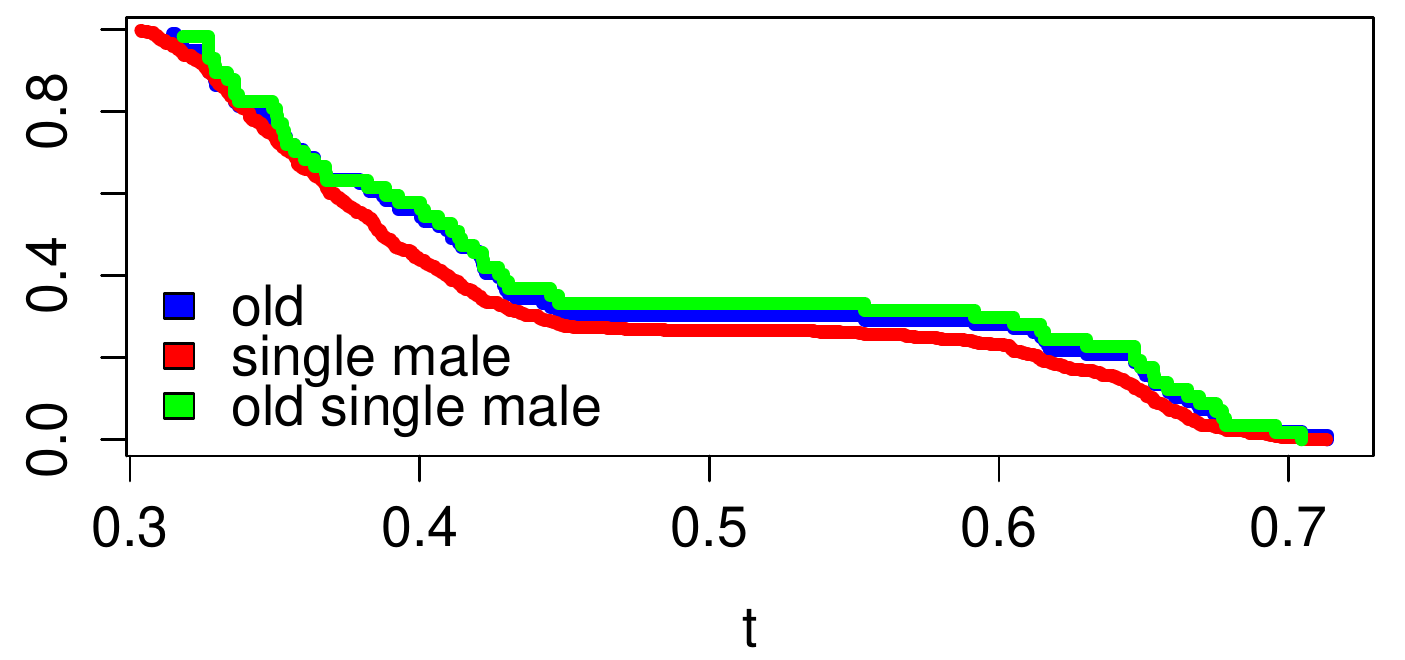}
\caption{Subgroup discrimination: plots reporting a threshold $t \in [0,1]$ on the X-axis, and the fraction of tuples having $gds^-() \geq t$ on the Y-axis. The top plot is from \textsf{Adult}, while the bottom is from \textsf{German credit}. \label{fig:multipleccdf}}
%\vspace{1mm}
\end{figure}

\subsection{Comparison with prior art}\label{subsec:exp_comp}

In this section, we discuss examples in which our causation-based method draws different conclusions from the correlation-based methods presented in \cite{peder2008,pederruggi2009,RPT2010} using the same datasets and  the same protected groups \footnote{We could not compare with~\cite{Mancuhan2014} due to repeatability issues.}.

\begin{figure}[h]
\centering
\begin{small}
\begin{tabular}{c|cc|c}
\multicolumn{1}{c}{} &  \multicolumn{2}{c}{decision}\\ \cline{2-3}
\multicolumn{1}{c}{} & \multicolumn{1}{c}{-} & \multicolumn{1}{c}{+} \\
\hline
foreign\_worker=yes & $298$ & $667$ \ & \ $968$\\
foreign\_worker=no & $2$ & $30$ \ & \ $32$\\ \hline
& \ $300$  & $700$ \ & \ $1000$\\
\end{tabular}
\quad \quad \quad
\begin{tabular}{c}
%\noindent
\\
$p_1 = 298/968=0.307$\\
$p_2 = 2/32=0.0625$\\
$\mathit{RD} = p_1 - p_2 = 0.244$
\end{tabular}
\end{small}
\caption{Contingency table for \textsf{foreign\_worker} in the \textsf{German credit} dataset.\label{fig:ct1}}

\vspace{3mm}

\begin{small}
\begin{tabular}{c|cc|c}
\multicolumn{1}{c}{} &  \multicolumn{2}{c}{decision}\\ \cline{2-3}
\multicolumn{1}{c}{} & \multicolumn{1}{c}{-} & \multicolumn{1}{c}{+} \\
\hline
race=black & $4119$ & $566$ \ & \ $4685$\\
race$\neq$black & $33036$ & $11121$ \ & \ $44157$\\ \hline
& \ $37155$  & $11687$ \ & \ $48842$\\
\end{tabular}
\quad \quad \quad
\begin{tabular}{c}
%\noindent
\\
$p_1 = 4119/4685=0.879$\\
$p_2 = 33036/44157=0.748$\\
$\mathit{RD} = p_1 - p_2 = 0.13$
\end{tabular}
\end{small}
\caption{Contingency table for \textsf{race\_black} in the \textsf{Adult} dataset.\label{fig:ct2}}

\vspace{-3mm}
\end{figure}

 The first example involves the \textsf{foreign\_worker} group from \textsf{German Credit} dataset, whose contingency table is reported in Figure~\ref{fig:ct1}. Following the approaches of  \cite{peder2008,pederruggi2009,RPT2010}  the \textsf{foreign\_worker} group results strongly discriminated. In fact Figure~\ref{fig:ct1} shows an $RD$ value (\emph{risk difference}) of 0.244 which is considered a strong signal: in fact $RD > 0$ is already considered discrimination \cite{RPT2010}.

 %\pagebreak

 However, we can observe that the \textsf{foreign\_worker} group is per se not very significant, as it contains 963 tuples out of 1000 total. In fact our causal approach does not detect any discrimination with respect to \textsf{foreign\_worker} which appears as a disconnected node in the \SBCN.

  The second example is in the opposite direction. Consider the \textsf{race\_black} group from \textsf{Adult} dataset whose contingency table is shown in Figure~\ref{fig:ct2}. Our causality-based approach detects a very strong signal of discrimination ($ds^-() = 0.994$), while the approaches of~\cite{peder2008,pederruggi2009,RPT2010} fail to discover discrimination against black minority when the value of minimum support threshold  used for extracting classification rules is more than $10\%$. On the other hand, when such minimum support threshold is kept lower, the number of extracted rules might be overwhelming. %Moreover, the value of $\mathit{RD}$ is not very strong, while in our method the discrimination reported is strong, regardless of the small size of the black population contained in the dataset.

Finally, we turn our attention to the famous example of false-positive discrimination case happened at Berkeley in 1973, that we discussed in Section 1.
Figure~\ref{fig:SBCNBAD} presents the \SBCN\ extracted by our approach from \textsf{Berkeley Admission Data}. Interestingly, we observe that there is no direct edge between node \textsf{sex\_Female} and \textsf{Admission\_No}. And \textsf{sex\_Female} is connected to node \textsf{Admission\_No} through nodes of \textsf{Dep\_C}, \textsf{Dep\_D}, \textsf{Dep\_E}, and \textsf{Dep\_F}, which are exactly the departments that have lower admission rate. By running our random walk-based methods over \SBCN\, we obtain the value of $1$ for the score of explainable discrimination confirming that apparent discrimination in this dataset is due the fact that women tended to apply to departments with lower rates of admission.

\begin{figure}[t!]
%\vspace{-3mm}
\centering
\includegraphics[width=\linewidth]{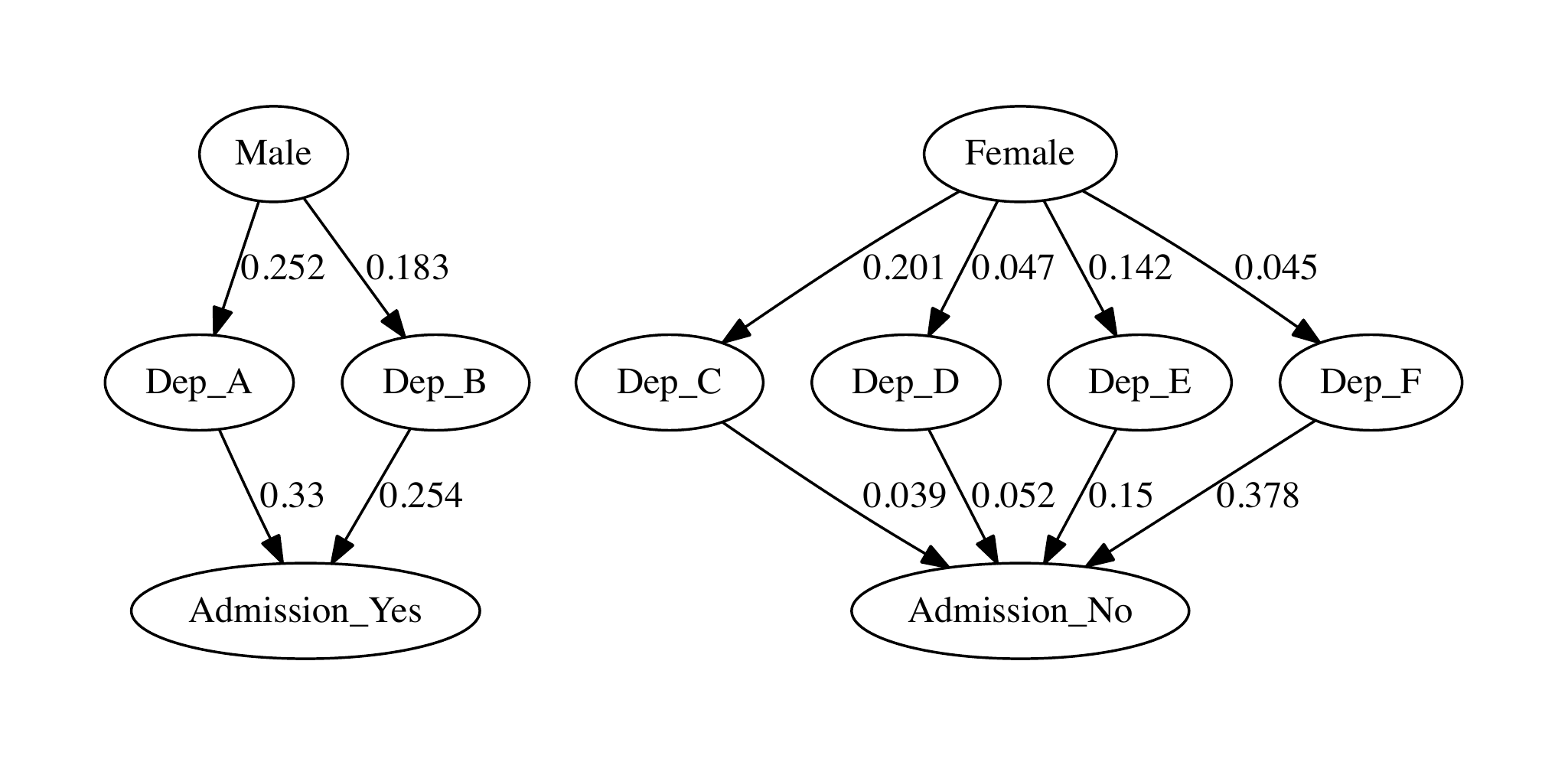}
\caption{The \SBCN\ constructed from \textsf{Berkeley Admission Data} dataset.\label{fig:SBCNBAD}}
\end{figure}

Similarly, we observe that there is no direct edge between node \textsf{sex\_Male} and \textsf{Admission\_Yes}. And \textsf{sex\_Male} is connected to node \textsf{Admission\_Yes} through nodes of \textsf{Dep\_A}, and \textsf{Dep\_B}, which are exactly the departments that have higher admission rate. By running our random walk-based methods over \SBCN\, we obtain the value of $1$ for the score of explainable discrimination confirming that apparent favoritism towards men is due to the fact that men tended to apply to departments with higher rates of admission.

\begin{figure}[h]
\centering
\begin{small}
\begin{tabular}{c|cc|c}
\multicolumn{1}{c}{} &  \multicolumn{2}{c}{decision}\\ \cline{2-3}
\multicolumn{1}{c}{} & \multicolumn{1}{c}{-} & \multicolumn{1}{c}{+} \\
\hline
gender=female & $1278$ & $557$ \ & \ $1835$\\
gender=male & $1493$ & $1158$ \ & \ $2651$\\ \hline
& \ $2771$  & $1715$ \ & \ $4486$\\
\end{tabular}
\quad \quad \quad
\begin{tabular}{c}
%\noindent
\\
$p_1 = 1278/1835=0.696$\\
$p_2 = 1493/2651=0.563$\\
$\mathit{RD} = p_1 - p_2 = 0.133$
\end{tabular}
\end{small}
\caption{Contingency table for \textsf{female} in the \textsf{Berkeley Admission Data} dataset.\label{fig:BADct}}
\end{figure}

However, following the approaches of ~\cite{peder2008,pederruggi2009,RPT2010}, the contingency table shown in Figure~\ref{fig:BADct} can be extracted from \textsf{Berkeley Admission Data}. As shown in Figure~\ref{fig:BADct}, the value of $\mathit{RD}$ suggests a signal of discrimination versus women. 

This highlights once more the pitfalls of correlation-based approaches to discrimination detection and the need for a principled causal approach, as the one we propose in this paper.

\section{Conclusions} \label{sec:conclusions}
\enlargethispage*{2\baselineskip}
Discrimination discovery from databases is a fundamental task in understanding past and current trends of discrimination, in judicial dispute resolution in legal trials, in the validation of micro-data before they are publicly released.
While discrimination is a causal phenomenon, and any discrimination claim requires to prove a causal relationship, the bulk of the literature on data mining methods for discrimination detection is based on correlation reasoning.

In this paper, we propose a new discrimination discovery approach that is able to deal with different types of discrimination in a single unifying framework.  It is the first discrimination detection method grounded in probabilistic causal theory.
We define a method to extract a graph representing the causal structures found in the database, and then we propose several random-walk-based methods over the causal structures, addressing a range of different discrimination problems.

Our experimental assessment confirmed the great flexibility of our proposal in tackling different aspects of the discrimination detection task, and doing so with very clean signals, clearly separating discrimination cases.

\bigskip

\noindent \textbf{Repeatability.} Our software together with the datasets used in the experiments are available at \url{http://bit.ly/1GizSIG}.


\end{document}